\documentclass[12pt,english]{article}
\usepackage{amsmath}
\usepackage{amsthm}
\usepackage{libertine}
\usepackage[libertine]{newtxmath}
\usepackage[T1]{fontenc}
\usepackage{textcomp}
\usepackage[latin9]{inputenc}
\usepackage{babel}
\usepackage{array}
\usepackage{prettyref}
\usepackage{float}
\usepackage{enumitem}
\usepackage{multirow}
\usepackage{graphicx}
\usepackage[letterpaper]{geometry}
\usepackage{fancyhdr}
\usepackage{setspace}
\usepackage{wasysym}
\usepackage[authoryear]{natbib}
\usepackage[pdfusetitle,
 bookmarks=true,bookmarksnumbered=false,bookmarksopen=false,
 breaklinks=false,pdfborder={0 0 1},backref=false,colorlinks=false]
 {hyperref}
\hypersetup{
 urlcolor=darkslateblue}

\makeatletter

\providecommand{\tabularnewline}{\\}
\floatstyle{ruled}
\newfloat{algorithm}{tbp}{loa}
\providecommand{\algorithmname}{Algorithm}
\floatname{algorithm}{\protect\algorithmname}

\theoremstyle{plain}
\newtheorem{thm}{\protect\theoremname}
\theoremstyle{plain}
\newtheorem{assumption}[thm]{\protect\assumptionname}
\theoremstyle{plain}
\newtheorem{lem}[thm]{\protect\lemmaname}
\theoremstyle{plain}
\newtheorem{cor}[thm]{\protect\corollaryname}
\theoremstyle{plain}
\newtheorem*{thm*}{\protect\theoremname}

\usepackage{fancyvrb}
\usepackage{graphicx}
\usepackage{tabularx}
\usepackage{color}
\usepackage{bbm}

\definecolor{darkslateblue}{rgb}{0.28, 0.24, 0.55}
\usepackage{lscape}
\usepackage{pdflscape}
\usepackage{float}
\usepackage{graphicx}
\usepackage{multicol}
\usepackage{booktabs}
\usepackage{multirow}
\usepackage{hhline}
\usepackage{adjustbox}
\usepackage[flushleft]{threeparttable}
\usepackage{algorithm,algpseudocode}
\usepackage{enumitem}
\setlist{nolistsep}
\usepackage[T1]{fontenc} 

\ifdefined\showcaptionsetup
 \PassOptionsToPackage{caption=false}{subfig}
\fi
\usepackage{subfig}
\makeatother

\providecommand{\assumptionname}{Assumption}
\providecommand{\corollaryname}{Corollary}
\providecommand{\lemmaname}{Lemma}
\providecommand{\theoremname}{Theorem}

\begin{document}
\begin{spacing}{1.0}
\title{Optimal Comprehensible Targeting\thanks{Email: \protect\href{mailto:wwz@wharton.upenn.edu}{wwz@wharton.upenn.edu}.
Earlier work on this project started with my dissertation from which
this project takes its title. I thank my dissertation advisors Sanjog
Misra, G\"unter J. Hitsch, Pradeep K. Chintagunta, Tengyuan Liang, and
Avner Strulov-Shlain for their encouragement and support. This paper
greatly benefited from discussions with Max H. Farrell, Yuexi Wang,
Karthik Srinivasan, James W. Kiselik, Olivia R. Natan, Kevin Lee,
Reuben Bauer, and Benedict Guttman-Kenney. I thank conference participants
at World Congress of the Econometric Society 2025, IMS ICSDS 2024,
ESIF on Economics and AI+ML 2024, ASA Joint Statistical Meetings 2024,
and NBER Summer Institute 2024, as well as seminar participants at
Columbia, Northwestern, Stanford, Berkeley, UCLA, University of Rochester,
Cornell, University of North Carolina, Duke, University of Pennsylvania,
Imperial College London, London Business School, Boston University,
University of Illinois, and Temple University for their helpful comments. }}
\author{Walter W. Zhang\\
{\normalsize\emph{The Wharton School, University of Pennsylvania}}}
\date{\today}
\maketitle
\begin{abstract}
Developments in machine learning and big data allow firms to fully
personalize and target their marketing mix. However, data and privacy
regulations, such as those in the European Union (GDPR), incorporate
a ``right to explanation,'' which is fulfilled when targeting policies
are comprehensible to customers. This paper provides a framework for
firms to navigate right-to-explanation legislation. First, I construct
a class of comprehensible targeting policies that is represented by
a sentence. Second, I show how to optimize over this class of policies
to find the profit-maximizing comprehensible policy. I further demonstrate
that it is optimal to estimate the comprehensible policy directly
from the data, rather than projecting down the black box policy into
a comprehensible policy. Third, I find the optimal black box targeting
policy and compare it to the optimal comprehensible policy. I then
empirically apply my framework using data from a price promotion field
experiment from a durable goods retailer. I quantify the cost of explanation,
which I define as the difference in expected profits between the optimal
black box and comprehensible targeting policies. Compared to the black
box benchmark, the comprehensible targeting policy reduces profits
by 7.5\% or 23 cents per customer.  

\vspace{0.5cm}

\begin{singlespace}
\end{singlespace}
\end{abstract}
\noindent\end{spacing}\thispagestyle{empty}

\newpage{}

\lhead{}

\rhead{}

\renewcommand{\headrulewidth}{0pt}\setstretch{1.15}

\section{Introduction}

\setcounter{page}{1}

Black box algorithms dominate marketing decisions today \citep{Katsov2017}.\footnote{A black box algorithm is one that users can observe the outcomes,
but the internal mechanisms remain opaque.} These algorithms are fast, personalized, and, most importantly, designed
to maximize profits for the firm. Modern algorithms apply profit-maximizing
marketing decisions at the individual level. At the extreme, customers
can face their own marketing mix. 

However, full personalization of the marketing mix may harm the firm.
Consider the following example in promotions management: two customers
go to check out but only one customer is given a 20\%-off promotion.
To justify the exclusion of the other customer, the firm's sales representative
cannot simply say it was profit-maximizing to do so. The representative
instead needs to provide the customer a comprehensible explanation
for why she did not receive the promotion. If a suitable explanation
cannot be found or understood, then the excluded customer may feel
slighted by the firm \citep{Dietvorst2022}.

More generally, customers and firms need explanations of algorithms.
Customers may desire an explanation either to understand the algorithm
\citep{Dietvorst2015}, or to learn from the algorithmic decision
\citep{VanOsselaer2000}\textemdash that is, customers want to know
why certain customers receive a promotion and the criteria for receiving
future promotions.

Recognizing such consumer needs as compelling interests, the General
Data Protection Regulation (GDPR) includes a ``right to explanation''
\citep{EuropeanCommission2016}: Firms need to fully explain their
algorithmic decisions to their consumers or they are subject to hefty
fines. As the AI Act \citep{EuropeanCommission2021} rolls out in
Europe and regulatory measures like the California Consumer Privacy
Act (CCPA) emerge in other jurisdictions, the scope of regulatory
oversight around right-to-explanation legislation is set to expand.
Beyond fulfilling statutory requirements, firms also need an explanation
of the algorithm to provide long-term brand equity by addressing customer
needs, making their algorithmic decisions easier to justify by human
representatives, and allowing themselves to self-diagnose their own
algorithmic decision making.

To provide the explanations that customers and firms need, firms can
adopt a sentence-length, comprehensible targeting policy, which must
be more than just a simple explanation of a black-box model's outcome.\footnote{The literature is divided as to whether local explanations of the
black box algorithms' outcomes are sufficient for regulators \citep{Edwards2017,Gillis2019}.} I argue a comprehensible policy must satisfy three conditions: it
must be \emph{transparent}, \emph{complete}, and \emph{conversational}.
Transparency means the targeting policy's logic is innately understandable
and easy to parse. Completeness means this transparent explanation
is globally identical to the policy being implemented; the explanation
is the algorithm, not a post hoc summary. Conversational means its
meaning is simple enough to be readily conveyed in day-to-day conversation.

Here is an example of a comprehensible policy: ``Target a customer
with a promotion if she has not bought in the last thirty days and
lives in Chicago.'' This policy is transparent because its logic
is expressed as a simple, human-readable sentence. It is also complete
because that sentence is the entire policy: there are no hidden weights
and everyone is given the same policy. It is conversational because
the sentence is concise: it only uses two conditional clauses.

This paper provides a framework for firms to design and analyze optimal
comprehensible targeting policies. Not only are these policies comprehensible
to customers, regulators, and the firm's own representatives, but
they are optimal in that they are profit-maximizing. My framework
allows firms to quantify the profit differences from implementing
optimal comprehensible policies as opposed to black box policies and
to compare the two targeting policies. There are three components
in my framework: (1) constructing a class of comprehensible policies,
(2) finding the optimal comprehensible policy, and (3) comparing it
to established the black box algorithms. 

First, I construct a class of comprehensible targeting policies. The
structure of the these comprehensible policies is motivated from the
philosophy and explainable AI literatures.\footnote{The philosophy literature notes that understanding is based on contrastive
causation and counterfactual inference \citep{Lipton1990}. These
two concepts are built into the constructed class of comprehensible
policies. The explainable AI literature suggests that sentences with
more than five constitutive clauses are not commonly understandable
and are not conversational \citep{Miller2019a}. } These comprehensible policies are targeting policies that consist
of a sentence formed from conditional clauses joined by logic operators
and I focus on shorter sentences to ensure they are conversational.
To be clear, this class of comprehensible policies is a possible class
of interpretable models that can be used in the framework and while
it is by no means the only definition of comprehensibility, it represents
a practical and reasonable one. I adopt this specific sentence-based
structure because it provides a clear, verifiable format that is transparent,
complete, and conversational. I further show how I can feature engineer
these clauses from standard database marketing datasets. 

Second, I find the optimal comprehensible policy that maximizes firm
profits. The components of the comprehensible policy over which I
optimize are the conditional clauses and the logic operators that
combine the clauses. I show how to solve the optimization problem
using brute force and greedy algorithms to learn the optimal policy
for a given sentence length. Then, I demonstrate how to conduct inference
around the expected profits generated from the optimal comprehensible
policy. 

Third, I evaluate the performance of the optimal comprehensible policy
against established black box targeting algorithms. This black-box
comparison group includes three approaches: (1) indirect methods that
focus on learning outcome variables, (2) direct methods that learn
treatment effects, and (3) policy learning methods that learn the
optimal targeting policy from the data \citep{Fernandez-Loria2023}.
I offer benchmarks for all three in the empirical application and
choose the best performing one for comparison. This benchmark represents
what a state-of-the-art firm would use it were unfettered by comprehensibility
constraints.

Using my proposed framework, I perform a cost analysis as firms adopt
comprehensible targeting policies to abide by right-to-explanation
legislation. I first compare whom the optimal black box and the optimal
comprehensible policies target. I then compute the profit difference
between the two policies. Since the comprehensible policy is less
expressive and personalized than the black box policy, the comprehensible
policy should be less profitable than the black box policy. I call
this difference in profits the \emph{cost of explanation} that the
firm faces if it adopts the comprehensible policy.

I implement this cost analysis with an empirical application in promotions
management. I use the dataset from \citet{Ni2012} as a case study
of the framework. The dataset includes a randomized control trial
of a \$10-off promotion randomly mailed to 176,961 households for
a durable goods retailer. The outcome of interest is sales during
the promotional period of December 2003. First, I show how to generate
and engineer different comprehensible policies from this RFM dataset.
Second, I find the optimal comprehensible policy and compare it to
the black box targeting policy. Third, I construct the black box benchmark
and find that combining policy learning with deep neural nets performs
the best in this application. These three steps let me document (1)
how the comprehensible targeting policy differs in whom it targets
compared to the black box targeting policy and (2) the firm's cost
of explanation if it implemented the optimal comprehensible policy
instead of the black box policy. 

Comparing the optimal comprehensible policy to the best-performing
black box policy, I find that the optimal comprehensible policy does
not systematically overtarget or undertarget; it appears to be limited
in its ability to capture customer heterogeneity due to the comprehensibility
constraint. In the empirical application, the optimal comprehensible
policy with three clauses targets those who spend a lot during the
holiday period but not in the spring, or who spent less than average
in prior holiday period.\footnote{More specifically, the optimal comprehensible policy is ``Target
a customer if she spent in the top half of spenders who spend during
Christmas over the last two years \emph{and} did not spend among the
top half of spenders in spring over the last two years \emph{or} is
among the bottom half of spenders during last year's holiday promotion.''} I further find that the cost of explanation is $23$ cents per person,
which constitutes a $7.5\%$ loss in profits compared to the optimal
black box policy and a $38$ cents per person (or $16\%$) gain in
profits over a blanket targeting policy. 

These results are of substantive interest. I provide an exercise where
I benchmark the GDPR penalty to the loss that the firm faces if it
moves away from the black box policy to comply with the right-to-explanation
legislation. In the empirical application, if I assume a basis of
10 million customers, the implied $23$-cent cost of explanation per
person leads to an expected loss of $\$2.3$ million per month. For
a GDPR penalty of around $\$20$ million that is enforced at a $10\%$
rate, then the expected penalty is $\$2$ million. The firm thus may
decide not to comply, depending on its assessment of the enforcement
rate and its appetite for regulatory risk. From the regulator's perspective,
it may consider raising the penalty or the enforcement rate to guarantee
compliance.\footnote{In fact, the EU\textquoteright s AI Act (2024) introduces higher fines
(maximum of 7\% of global revenue or 35 million Euros) when compared to
GDPR (2018) fines (maximum of 4\% of global revenue or 20 million Euros).}

Moving from the empirical application back to theory, I show that
directly forming the optimal comprehensible policy from the data is
more profitable than ex post projecting down the black box policy
to a comprehensible policy. The latter procedure is motivated by the
explainable AI (XAI) literature. Methods in XAI often provide a locally
approximative model of the black box in order to shed light upon the
black box's decisions \citep{Biran2017a,Miller2019a,Mothilal2019,Rai2020,Senoner2022}.
I show that the direct approach theoretically generates more profits
than the ex post approach and validate this using the empirical application.
I also show how I can conduct inference on the ex post comprehensible
policy by recentering the empirical process results from \citet{Kitagawa2018}.
Based on these results, firms should directly form the comprehensible
policy from the data instead of first finding the optimal black box
and then projecting it down.

This paper builds on many different extant literatures. The implementation
of algorithms in decision making and their effects are well documented
\citep{Kleinberg2015,Kleinberg2017}. In marketing, algorithmic decisions
are often made by forming optimal targeting policies in various domains
\citep{Ascarza2018,Chintagunta2023,Ellickson2022,Hitsch2018,Karlinsky-Shichor2019,Rossi1996,Simester2020,Smith2022,Yoganarasimhan2020,Zhang2022}
and the literature is reviewed in \citet{Rafieian2022}. The comprehensible
policy class builds on the interpretable AI literature where expert
systems and logic rules are designed to be interpretable by human
agents \citep{Angelino2018,Cawsey1991a,Cawsey1992a,Cawsey1993a,Rudin2019,Wang2017,Weiner1980a}. 

Recent work on interpretable and explainable AI in social science
has recently examined its technical performance, market impact, and
regulatory implications. In terms of performance, recent research
in marketing study how to incorporate interpretable structure into
black box models to aid their performance. \citet{Wang2022} build
a hybrid interpretable model that is shown to be more effective than
pure black-box methods for predicting marketing responses. \citet{Fong2021}
incorporate interpretable structure from music theory in a deep neural
network to help performance. \citet{Wang2025} demonstrate that adding
contrastive explanations to recommendation systems can improve recommendations.
In terms of market impact, \citet{Mohammadi2024} and \citet{Wang2023}
both analytically analyze the effects of AI transparency on firms,
consumers, and customer welfare. Focusing on regulation, \citet{Lambin2023}
provide a game-theoretic analysis of firm and regulator behavior under
right-to-explanation laws. 

Methodologically, this paper deviates from the interpretable AI literature
in that the comprehensible policy is much simpler and is conversational\textemdash a
firm's representative is able to fully state the targeting policy
in a sentence. Unlike typical interpretable models that offer a \textquotedbl dictionary\textquotedbl{}
of rules, requiring users to look up their specific case, my approach
provides a single, simple sentence that applies universally to all
users. This distinction is crucial: Regulators cite their inability
to diagnose complex algorithms and their explanations as a barrier
to oversight \citep{Edwards2017,Kleinberg2018}. My framework for
finding and evaluating optimal comprehensible policies provides a
potential solution for firms to comply with regulatory demands while
balancing profitability. 

Substantively, this paper provides a quantitative framework for estimating
how right-to-explanation laws impact profitability for individual
firms. Prior work has theorized about regulatory impact on firms;
I show precisely how to measure the cost of implementing comprehensible
targeting policies and how firms can optimize these policies under
regulatory constraints. This quantification offers concrete, measurable
trade-offs, with my empirical application demonstrating a profit loss
of $7.5\%$ for a targeted marketing campaign. 

The rest of the paper is organized as follows: Section \ref{sec:Framework}
provides an overview of the methodology and sets up the mathematical
notation. Section \ref{sec:Comprehensible-class} constructs a class
of comprehensible targeting policies and Section \ref{sec:Optimal-comprehensible-policy}
shows how to find the optimal comprehensible policy. I show that generating
a comprehensible policy from projecting down a black box is less profitable
than learning the policy directly from the data in Section \ref{sec:Ex-post-policy}.
I provide the empirical application in Section \ref{sec:Empirical-application}:
in Section \ref{subsec:Targeting-Policy-Difference}, I document the
differences in whom the black box and the optimal comprehensible targeting
policy target and in Section \ref{subsec:Targeting-Policy-Profit-Difference},
I quantify the cost of explanation. I then discuss what firms should
consider when deciding to implement comprehensible policies in Section
\ref{sec:Discussion} and quantify the impact of GDPR's right-to-explanation
regulation on the firm in the empirical application in Section \ref{subsec:GDPR-Calibration}.
I conclude in Section \ref{sec:Conclusion}.

\section{Framework\label{sec:Framework}}

In this section, I provide an overview of the general methodology
of forming and evaluating comprehensible policies. Figure \ref{fig:Methodological-overview}
summarizes the technical framework, which I unpack in Section \ref{subsec:Framework-overview}.
I introduce the mathematical framework in Section \ref{subsec:Mathematical-framework}
and discuss the three main approaches for learning targeting policies
in Section \ref{subsec:Learning-optimal-targeting}. A key methodological
theme for find both the black box and optimal comprehensible policies
is policy learning, where the optimal targeting policy is directly
learned from the data by maximizing profits.

\subsection{Framework overview \label{subsec:Framework-overview}}

Figure \ref{fig:Methodological-overview} illustrates the perceived
trade-off between comprehension and expected short-term profits. Comprehension
can be considered as $1/(\text{Model Complexity})$ and can be made
mathematically rigorous by using the targeting policy's Vapnik-Chervonenkis
dimension or Rademacher complexity as a proxy for model complexity.

As targeting policies become more comprehensible, they become less
personalized and cannot capture customer heterogeneity as well as
incomprehensible black box methods. Each dot on the figure represents
a specific targeting policy. Targeting policies that are simple enough
to be explained in a conversation are to the right of the dashed line.
For example, a blanket targeting policy would be very comprehensible
and conversational but not profitable, and it would be represented
by a point on the bottom right side of the figure near the horizontal
axis.

The proposed methodological framework evaluates and compares comprehensible
targeting policies to black box targeting policies. There are three
components for the framework: (1) constructing a class of comprehensible
policies, (2) finding the optimal profit-maximizing comprehensible
policy, and (3) forming the black box policy benchmark. I now describe
the three pieces to form the framework while using Figure \ref{fig:Methodological-overview}
as a guide.

First, I construct a class of comprehensible policies ($\mathcal{F}_{\text{Comp}}$)
in Section \ref{sec:Comprehensible-class} and represent it with the
red curve on the right side of the figure. I show how to generate
the clauses for the comprehensible policies from marketing data and
then form the policies. For this function class of sentences, the
Vapnik-Chervonenkis dimension is determined by the number of clauses
and represents the model complexity (Appendix Section \ref{sec:Decision-Trees-and-Sentences}).
To trace out the red curve, a targeting policy ``Target a customer
if she has not bought in the last thirty days'' would be a point
on the very right side of the red curve as it is quite comprehensible
but is limited in its profitability. The targeting policy ``Target
a customer if she has not bought in the last thirty days and lives
in Chicago'' will also be on the curve but is to the left of the
former example. Since the sentence has two clauses, it will be less
comprehensible but more profitable. The red curve also trends down
after a point to represent overfitting; a targeting policy that contains
twenty clauses can overfit and is also not conversational.

Second, I show how to optimize over this class of comprehensible policies
in Section \ref{sec:Optimal-comprehensible-policy} by using brute
force and greedy algorithms for a given length of clauses. I find
the point on the red curve that generates the highest profits and
denote this targeting policy as $d_{\text{Comp}}^{*}(x)$ which generates
expected profits $\Pi_{\text{Comp}}^{*}$. I also show that finding
the optimal comprehensible policy directly from the data is more profitable
than finding it by projecting down the black box in Section \ref{sec:Ex-post-policy}. 

Third, I focus on the class of black box policies ($\mathcal{F}_{BB}$)
and represent this class of functions on the left side of the figure
with the blue curve. Black box functions, such as deep neural nets,
possess a uniform approximation property that allows them to approximate
any function \citep{Goodfellow-et-al-2016}. As the blue curve moves
from right to left it trades off profitability for comprehensibility.
For example, a DNN with only one hidden layer with a handful of nodes
would be a point on the right side of the blue curve while a larger
DNN would be on the left side of the curve. The curve trends down
after a certain point to represent overfitting when the DNN is too
complex. The optimal black box policy function $d_{BB}^{*}(x)$ is
represented as a point on the blue curve in the figure and it generates
expected profits of $\Pi_{BB}^{*}(x)$.

Using the empirical application, I then compare the two targeting
policies from the optimal black box $d_{BB}^{*}(x)$ and the optimal
comprehensible policy $d_{\text{Comp}}^{*}(x)$ to examine the differences
in the targeting policies in Section \ref{subsec:Targeting-Policy-Difference}.
I then document the cost of explanation, or the difference in profits
from the optimal black box to the optimal comprehensible policy ($\Delta\Pi=\Pi_{BB}^{*}-\Pi_{\text{Comp}}^{*}$),
in Section \ref{subsec:Targeting-Policy-Profit-Difference}.

This framework allows me to substantively evaluate the firm's economic
losses when it follows right-to-explanation laws for its targeting
policy. I compare the cost of explanation to the expected GDPR penalty
and this exercise allows me to evaluate the penalty's impact on the
firm in the empirical application. I discuss how these losses play
a role for managerial decision making in Section \ref{sec:Discussion}.

\subsection{Mathematical framework \label{subsec:Mathematical-framework}}

I define $(X,W,Y)$ as the data tuple of the covariates, the treatment,
and the outcome variable respectively. I observe this tuple for each
individual $i$ and assume the data is $i.i.d.$ for $n$ observations.
I consider a binary treatment $W\in\{0,1\}$ and the data $X$ has
dimension $p$. I further define $Y(1),Y(0)$ as the potential outcomes
for the binary treatment with observed outcome variable $Y=WY(1)+(1-W)Y(0)$.

To provide a concrete example, I can think of $Y$ as sales, $W$
as a promotional mailing, and $X$ as consumer characteristics. The
firm cares about maximizing expected profits from its promotions management
campaign and profits are a function of the sales.

Consider the standard inverse propensity weighted profit estimator
for a given policy function $d:\mathbb{R}^{p}\to\{0,1\}$ that maps
the covariates to the targeting rule,
\begin{equation}
\hat{\Pi}(d)=\sum_{i=1}^{n}\frac{W_{i}}{e(x_{i})}\pi_{i}(1)d(x_{i})+\frac{1-W_{i}}{1-e(x_{i})}\pi_{i}(0)(1-d(x_{i})).\label{eq:IPWE-Profits}
\end{equation}
This profit estimator is an unbiased estimator for the profits from
the targeting policy 
\begin{equation}
\Pi(d)=\sum_{i=1}^{n}\pi_{i}(1)d(x_{i})+\pi_{i}(0)(1-d(x_{i}))\label{eq:True-Profits}
\end{equation}
 where $\pi_{i}(1)=mY_{i}(1)-c,\pi_{i}(0)=mY_{i}(0)$ are the counterfactual
profit values when individual $i$ is respectively targeted and not
targeted, $e(x_{i})=P(W_{i}\mid X=x_{i})$ is the propensity score,
and $m,c$ respectively represent the profit margins and the cost
of issuing the treatment.\footnote{The inverse propensity weighted profit estimator is also known as
the Horvitz-Thompson profit estimator \citep{Imbens2015}, and $E[\hat{\Pi}]=\Pi$
under the standard assumptions of unconfoundedness, overlap, and SUTVA
\citep{Hitsch2018}.}

From Equation \prettyref{eq:True-Profits}, I see that the optimal
policy function is 
\begin{equation}
d^{*}=\mathbf{1}\{\pi_{i}(1)>\pi_{i}(0)\}=\mathbf{1}\{m(Y_{i}(1)-Y_{i}(0))>c\}\label{eq:True-Policy-Opt}
\end{equation}
 in which individual $i$ is targeted if and only if her counterfactual
profits are higher under that treatment assignment, and where $\mathbf{1}\{\cdot\}$
represents the indicator function.

I do not observe $Y_{i}(1),Y_{i}(0)$ for each individual $i$ because
of the fundamental problem of causal inference, so I cannot form $Y_{i}(1)-Y_{i}(0)$
in Equation \ref{eq:True-Policy-Opt}. Instead, I want to find a representation
for the difference in the potential outcomes. 

I make a few assumptions to do so: First, I have a structural assumption
that customers have a utility function $u_{i}=\alpha(x_{i})+\beta(x_{i})W_{i},$
which maps to the outcome variable with function $Y_{i}=G(u_{i})+\epsilon_{i}$.
I further assume a linear function for $G(\cdot)$, $G(u_{i})=u_{i}$.
The assumption simplifies the problem to 
\begin{equation}
Y_{i}=\alpha(x_{i})+\beta(x_{i})W_{i}+\epsilon_{i},\label{eq:Outcome-Equation}
\end{equation}
where I can interpret $\alpha(x_{i})$ as the baseline sales if customer
$i$ was not given a treatment and $\beta(x_{i})$ as the incremental
effect of issuing the treatment on sales for customer $i$. I allow
the intercept term $\alpha(x_{i})$ and the coefficients $\beta(x_{i})$
to depend on the individual's pretreatment covariates, which provide
both heterogeneity in the coefficients and the ability to forecast
$\alpha(x_{i}),\beta(x_{i})$ for a customer with given $x_{i}$.

I make three additional assumptions to ensure I can recover $\beta(x_{i})$
as a causal effect of the treatment $W_{i}$ \citep{Imbens2015}.
These assumptions are typical in experimental settings. The first
two are the unconfoundedness and the overlap assumptions. Under a
properly run randomized control trial (RCT), these two assumptions
are satisfied. I then assume the stable unit treatment value assumption
(SUTVA) holds. This assumption effectively implies that there are
no spillover effects from sharing the promotional mailing in the running
example. I formally state the assumptions below.
\begin{assumption}
\label{assu:Unconfoundedness}(Unconfoundedness) The potential outcomes
$Y_{i}(1),Y_{i}(0)$ are statistically independent of the treatment
variable $W_{i}$ which is represented formally as $\{Y_{i}(1),Y_{i}(0)\}\perp W_{i}$.
\end{assumption}

\begin{assumption}
\label{assu:Overlap}(Overlap) The propensity score $e(x_{i})=P(W_{i}\mid X=x_{i})$
is bounded between zero and one which is represented formally as $0<e(x_{i})<1$.
\end{assumption}

\begin{assumption}
\label{assu:SUTVA}(SUTVA) The potential outcomes for any individual
do not vary with the treatment assignments for other individuals.
For each individual, and there are no different forms of treatments
that lead to different potential outcomes.
\end{assumption}

These three assumptions allow me to identify the conditional expectation
of potential outcomes and then map them to the observed data. I have
that
\begin{align*}
\alpha(x_{i}) & =E[Y_{i}(0)\mid X=x_{i},W=0]=E[Y_{i}\mid X=x_{i},W=0]\\
\alpha(x_{i})+\beta(x_{i}) & =E[Y_{i}(1)\mid X=x_{i},W=1]=E[Y_{i}\mid X=x_{i},W=1]\\
\beta(x_{i}) & =E[Y_{i}(1)\mid X=x_{i},W=1]-E[Y_{i}(0)\mid X=x_{i},W=0]\\
 & =E[Y_{i}\mid X=x_{i},W=1]-E[Y_{i}\mid X=x_{i},W=0]
\end{align*}
and $\beta(x_{i})$ is the heterogeneous treatment effect (HTE) or
the conditional average treatment effect (CATE) for treatment $W_{i}$. 

Circling back to the optimal policy $d^{*}$ in Equation \prettyref{eq:True-Policy-Opt},
I can learn the optimal policy function as a function of covariates
$x_{i}$, 
\begin{align}
d^{*}(x_{i})=d^{*}(\beta(x_{i})) & =\mathbf{1}\{m(E[Y_{i}(1)\mid X=x_{i},W=1]-E[Y_{i}(0)\mid X=x_{i},W=0])>c\}\label{eq:Policy-Function}\\
 & =\mathbf{1}\{\beta(x_{i})>\frac{c}{m}\},\nonumber 
\end{align}
where I use the conditional expectation of the potential outcomes
in place of the potential outcomes. 

\subsection{Learning the optimal targeting policy \label{subsec:Learning-optimal-targeting}}

Recent approaches to form the optimal targeting policy take a three-step
approach, a two-step approach, or use policy learning. The three-step
approach first estimates the conditional expectation functions $E[Y_{i}(1)\mid X,W=1]$
and $E[Y_{i}(0)\mid X,W=0]$ using regressions, then forms the treatments
effects $\hat{\beta}(x)=\hat{E}[Y_{i}(1)\mid X,W=1]-\hat{E}[Y_{i}(0)\mid X,W=0]$,
and lastly constructs the optimal policy $\hat{d}^{*}(x_{i})=\mathbf{1}\{\hat{\beta}(x_{i})>\frac{c}{m}\}$
following the plug-in rule (Equation \ref{eq:Policy-Function}). The
two-step approach first estimates the heterogeneous treatment effect
$\hat{\beta}(x)$ and then constructs the optimal policy $\hat{d}^{*}(x_{i})=\mathbf{1}\{\hat{\beta}(x_{i})>\frac{c}{m}\}$
from the plug-in rule. Policy learning learns the optimal policy $d^{*}(x_{i})$
in one-step from the data. 

Black box approaches exist for all three procedures. Three-step approaches
model conditional expectation functions using Random Forests, Lasso,
or Deep Neural Networks \citep{Hitsch2024}. Two-step approaches directly
model treatment effects, typically with Causal Forests \citep{Athey2019}.
Policy learning methods either use simple models like linear threshold
rules in \citet{Kitagawa2018} or construct a Policy Tree by projecting
Causal Forest estimates onto decision trees \citet{Athey2021}.

The theoretical foundation for these policy learning methods relies
on complexity constraints. \citet{Kitagawa2018} and \citet{Athey2021}
derive minimax regret rates for policy functions, restricting complexity
to ensure the square root of the Vapnik-Chervonenkis (VC) dimension
grows no faster than $\sqrt{n}$, to follow the fundamental theorem
of statistical learning \citep{Vapnik2000}. In practice, this constrains
policies to linear combinations or shallow trees of covariates. More
recent work replaces conservative minimax rates with semiparametric
inference, enabling more flexible functional forms. For example, \citet{Zhang2024}
uses deep neural networks to represent the policy function $\hat{d}^{*}(x)$
and conducts inference around optimal targeting profits.

In my framework, I benchmark these three black box approaches to establish
an optimal unconstrained baseline. The best-performing black box represents
the state-of-the-art targeting policy a firm would deploy absent comprehensibility
constraints, which provides a reference point for quantifying the
cost-of-explanation requirements.

\section{Forming a comprehensible class of policies\label{sec:Comprehensible-class}}

I now define comprehensible policies as policies expressible in natural
language sentences. I denote this class of policies as comprehensible
targeting policies. Here is an example of a comprehensible policy:
\[
\text{\textquotedblleft Target customer if she lives in Chicago and she has not bought in the last thirty days.\textquotedblright\ }
\]
The sentence has two conditional clauses ``if she lives in Chicago''
and ``{[}if{]} she has not bought in the last thirty days'' that
are linked by the ``and'' logic operator. Thus the customer will
be targeted if both clauses are true and will be not targeted otherwise.
This class of comprehensible policies is denoted by $\mathcal{F}_{\text{Comp}}$.

I describe what makes a targeting policy comprehensible in Section
\ref{subsec:What-is-comprehensibility}, construct the comprehensible
policies class in Section \ref{subsec:Comprehensible-policies}, and
show how I can engineer the clauses from standard recency, frequency,
and monetary (RFM) marketing data in Section \ref{subsec:Generating-clauses}.

\subsection{What is comprehensibility?\label{subsec:What-is-comprehensibility}}

The sentence, \textquotedblleft Target customer if she lives in Chicago
and she has not bought in the last thirty days\textquotedblright{}
is comprehensible in the sense it is \emph{transparent}, \emph{complete},
and \emph{conversational}. For transparency, sentences are the natural
medium that we communicate in. The firm can state the sentence to
the customer or regulator, and the sentence is easy to explain and
parse. The number of clauses captures the complexity of the targeting
policy and a sentence with more unique clauses will necessarily be
more complex. For completeness, the sentence is the exact targeting
policy being implemented by the firm; there is a one-to-one mapping
between the explanation provided to the implemented policy. 

While transparency might also suggest conversational, I emphasize
the conversational requirement to ensure these comprehensible policies
remain easy to parse. Drawing on \citet{Miller2019a}, I restrict
the complexity of the policy to be at most five clauses because longer
explanations can be difficult to parse. While a sentence with more
than five clauses may be grammatically correct, it becomes difficult
for a sales representative to articulate it to the customer. To operationalize
this conversational requirement, I mainly focus on comprehensible
targeting policies with at most five clauses. 

The comprehensible policies are constructed to embed contrastive explanations
\citep{Lipton1990}. The explainable AI literature uses the idea of
contrastive explanations as a core component of the explanation itself
\citep{Biran2017a,Halpern2005b,Halpern2005c,Miller2019a,Mothilal2019}.\footnote{In western philosophy, Spinoza holds the strongest position about
explanation, namely that everything is part of a causal chain of explanation.
While some explanations are difficult to find, Spinoza equates denying
these explanations with seeking refuge in ``the sanctuary of ignorance''
\citep{BSEthics}. This account refutes contrastive explanation: each
explanation is already determined by others. As a contemporary Spinoza
commentator writes, ``our place in the world is simply the way in
which we are explained by certain things and can serve to make intelligible\textemdash i.e.,
explain\textemdash certain other things'' \citep{DellaRocca2008}.} In the running example, a contrastive explanation means that since
customers are only targeted if they live in Chicago and have not bought
in the last thirty days, the negation of either conditional clause\textemdash ``if
they do not live in Chicago'' or ``have bought in the last thirty
days''\textemdash implies the customer would have not been targeted.
By providing the customer with the targeting rule and counterfactual
states in which the customer would not be targeted, the firm is able
to explain the comprehensible targeting in a manner that the customer
can understand.

Fundamentally, the choice of this class of comprehensible policies
is motivated by the emphasis on \emph{transparent}, \emph{complete},
and \emph{conversational} targeting policies that embed contrastive
explanations. This class of policies deviates from those proposed
in the interpretable AI literature, which suggests constructing long
rule lists \citep{Angelino2018,Rudin2019}, expansive decision trees
\citep{Weiner1980a}, or large flow charts \citep{Cawsey1992a}.
As a result, the comprehensible policy class will be simpler than
those used in interpretable AI: Large rule lists and decisions trees
may be transparent and complete, but they are not conversational.

\subsection{Comprehensible policies \label{subsec:Comprehensible-policies}}

I consider the class of sentences that consistent of clauses linked
by logic operators as the proposed class of policies. I call this
class of models \emph{comprehensible policies}. In essence, I want
to capture targeting rules that can be represented as a coherent and
grammatically correct sentences which say, 

\begin{align}
\text{\textquotedblleft Target if customer~} & \text{\ensuremath{i} has \textbf{this}\textquotedblright}\nonumber \\
\text{\textquotedblleft Target if customer~} & \text{\ensuremath{i} has \textbf{this} \emph{and} \textbf{that}\textquotedblright}\label{ex:Comprehensible-Sentence}\\
\text{\textquotedblleft Target if customer~} & \text{\ensuremath{i} has \textbf{this} \emph{or} \textbf{that} \emph{and} \textbf{not this}\textquotedblright}\nonumber 
\end{align}
 where ``\textbf{this}'' and ``\textbf{that}'' are two conditional
clauses described by covariates and are linked by the logic operators
``\emph{and},'' \emph{``or},'' and ``\emph{xor}''.\footnote{The operator ``xor'' represents the exclusive or which captures
A or B but not A and B.} The clauses can be negated so ``\textbf{this}'' can be formed into
``\textbf{not this}'' by prefixing the clause with the \emph{``not''}
operator.\footnote{The running example, ``Target customer if she lives in Chicago and
she has not bought in the last thirty days'' has two total conditional
clauses and they are linked by the logic operator ``and.'' The conditional
clauses themselves are ``she lives in Chicago'' and ``she has not
bought in the last thirty days.''} In this setup, the complexity of the explainable policy is defined
as the number of conditional clauses used and is denoted as $\ell$.

While comprehensible policies are simple to state, optimizing over
the sentences can be combinatorially difficult. If there are $k$
distinct possible clauses and three logic operators for a sentence
of length $\ell$ (that uses $\ell$ total clauses), then the total
possible combinations for that sentence is $3^{\ell-1}(2k)^{\ell}$,
which is exponential in the number of possible clauses. 

The comprehensible policies can be expressed in the most general form
in logic trees and decision trees. Logic trees take in binary covariates
and generate a sentence that uses the covariates as clauses and links
them using logic operators \citep{Schwender2010}. The elements of
the tree are the binary covariates and the logic operators are \emph{and},
\emph{or}, and \emph{not}. The logic trees can then be collapsed into
a sentence of the form in Example \ref{ex:Comprehensible-Sentence}.
Comprehensible policies are simpler versions of decisions trees and
a comprehensible policy with $\ell$ clauses can be represented by
a decision tree of depth $\ell$ but not vice versa.\footnote{Appendix Lemma \ref{lem:Comprehensible-Sentence-Decision-Tree} maps
comprehensible policies to decision trees and Appendix Corollary \ref{cor:Tree-to-Sentence}
shows that not all decision tree of depth $\ell$ can be mapped to
a comprehensible policy with $\ell$ clauses.} The procedure to generate corresponding decision trees from a comprehensible
policy is described in Appendix Section \ref{sec:Decision-Trees-and-Sentences}.

\subsection{Generating clauses \label{subsec:Generating-clauses}}

Comprehensible policies require clauses with binary values, and I
show how I can generate these clauses from standard RFM marketing
data in this section. The clauses depend on the data available to
the firm, and I assume the data are of standard tabular form in which
each column of the data is itself comprehensible. In the RFM setting,
a data column that describes sales in the past 12 months is assumed
to be comprehensible to mangers, customers, and regulators. 

Binary covariates in the data do not require further processing since
they can be represented by clauses directly. For example, a binary
covariate that is 1 if a customer is ``a new user'' and 0 otherwise
is represented by the clause ``is a new user.''

Categorical covariates are expanded to binary covariates and then
turned into the clauses. For example, the type of mobile phone that
a customer has can be converted into binary variables (``has an iPhone,''
``has an Android,'' and ``has a flip phone'') and these are converted
into clauses directly. The clauses from the categorical variables
are more expressive, as a clause like ``does not have a flip phone''
would capture the customer having either an ``iPhone'' or an ``Android''. 

For continuous covariates, I discretize them into three bins that
are each represented by a binary indicator. I implement this discretization
with RFM data in mind. To provide a concrete example, Figure \ref{fig:Generating-Clauses-from-Continuous-Variables}
shows the unconditional distribution for the past November sales covariate
from the empirical application in the upper panel and the conditional
distribution for nonzero past November sales in the lower panel. I
first note that $97\%$ of the observations are zero so I first construct
the clause ``has \emph{zero} past November sales.'' 

I then look at the conditional distribution of past November sales
for the customers that spend during November in the lower panel. The
median value is $\$129.99$ so I classify all customers above this
value to be ``high'' and the customers below this value and above
zero to be ``low.'' Their respective clauses will be ``has \emph{low}
past November sales among spenders'' and ``has \emph{high} past
November sales among spenders.'' I follow this procedure in the empirical
application for the RFM dataset. In other settings, researchers can
construct three clauses of low, medium, or high values by cutting
up the continuous variable at the empirical quantiles.\footnote{A more general approach can also discretize the continuous covariates
into smaller bins based on deciles or more precise quantiles but these
will significantly expand the total number of clauses to optimize
over.}

Naturally, data scientists can transform the data into different representations
or embeddings, but I consider the transformed variables largely to
not be comprehensible. At an extreme, a data scientist can estimate
heterogeneous treatment effects $\hat{\beta}(x_{i})$ using a black
box method. Then a sentence would naturally recover the optimal policy
function ($d^{*}(x_{i})$) that says, ``Target customer if she has
$\hat{\beta}(x_{i})>c/m$.'' Because the process of attaining $\hat{\beta}(x_{i})$
is opaque, this targeting policy is not transparent and thus not comprehensible.

\subsection{Why use the class of sentences?}

Sentences represent one of many possible comprehensible policy classes.
The interpretable AI literature offers alternatives ranging from decision
trees to lists of decision rules \citet{Rudin2019}. The critical
distinction is that sentences are \emph{transparent}, \emph{complete},
and \emph{conversational}. Other interpretable AI models typically
do not satisfy all three criterion.

For example, enumerating an entire dictionary of decision rules offers
a global explanation but is not conversational. Conversely, while
tracing an individual\textquoteright s path in the dictionary is conversational,
it yields only a local explanation rather than a global one. In contrast,
a single sentence is conversational and applies universally to all
customers, creating a one-to-one mapping between the explanation and
the implemented policy.

This global simplicity has important regulatory implications. Given
ongoing legal debate about what explanations satisfy right-to-explanation
requirements \citep{Edwards2017}, simple sentences should more readily
achieve compliance than complex decision trees or lists of decision
rules. Consequently, sentence-based policies represent a lower bound
on both profitability and regulatory action. Adding complexity through
decision trees or rule dictionaries may increase targeting profits
but simultaneously raises regulatory hurdles, as more complex explanations
face greater scrutiny. If even simple sentences (the most likely to
satisfy regulators) cannot achieve profitable compliance, then firms
face a tension between the profit gains from increasing complexity
and the regulatory costs it imposes.

\section{Optimal comprehensible policy\label{sec:Optimal-comprehensible-policy}}

I now develop methods to find optimal comprehensible policy among
the class of sentence-based policies from in Section \ref{sec:Comprehensible-class}.
I denote the class of comprehensible policies as $\mathcal{F}_{\text{comp}}^{\ell}$
and now make the dependence on the number of clauses $\ell$ in the
sentence explicit. Recall that $\ell$ represents the complexity of
the comprehensible policy: A sentence with more clauses can provide
a finer partition of the customer base. For example, $\mathcal{F}_{\text{comp}}^{2}$
represents this class of comprehensible policies that are spanned
by sentences with $\ell=2$ clauses. 

I now focus on optimizing over this class of policies for a given
$\ell$ value using policy learning. Specifically, I use the direct
empirical welfare maximization framework \citep{Kitagawa2018} where
I directly find the optimal comprehensible policy $d_{\text{comp}}^{*}(x)\in\mathcal{F}_{\text{comp}}^{\ell}$
that maximizes the sample profit estimator in Equation \ref{eq:IPWE-Profits},
\begin{equation}
d_{\text{comp}}^{*}(x)=\underset{d'\in\mathcal{F}_{\text{comp}}^{\ell}}{\arg\max}\ \hat{\Pi}(d').\label{eq:d_comp_opt}
\end{equation}
I propose two procedures that directly maximize profits over this
class of comprehensible policies. The first is brute force optimization,
which guarantees the globally optimal solution but can be computationally
intensive. The second is a greedy algorithm that is computationally
tractable but may result in a locally optimal solution.\footnote{While I develop brute force and greedy algorithms for this targeting
problem, related optimization challenges appear in other domains.
\citet{Hauser2010} use machine learning methods to identify disjunction-of-conjunction
decision rules, which are a similar logical structure, to predict
consumer consideration sets.} I then provide inference around the optimal comprehensible policy
by leveraging results from \citet{Kitagawa2018}.

\subsection{Brute force algorithm\label{subsec:Brute-force-optimization}}

I first propose the brute force optimization approach where I first
enumerate over all possible comprehensible policies and then choose
the one that yields the highest expected profits. Algorithm \ref{alg:Brute-force-optimization}
details the the brute force algorithm's steps.

\begin{algorithm}
\caption{Brute force optimization \label{alg:Brute-force-optimization}}

\textbf{Setup}: Number of clauses $\ell$:
\begin{enumerate}
\item Discretize the $p$ covariates into $q$ pieces to get the clauses
\item Combinatorially iterate through all targeting policies combinations
of the $pq$ clauses and logic operators
\item Choose the policy with $\ell$ clauses that maximizes profits
\end{enumerate}
\end{algorithm}

The problem of finding a $\ell=3$ clause sentence has the structure:
\[
\text{Target if \{A\} <and/or/xor> \{B\} <and/or/xor> \{C\}}
\]
where \{A\}, \{B\}, \{C\} are the clauses in the sentence and <and/or/xor>
are the logic operators. To find the optimal comprehensible policy
with three clauses, the brute force approach would first enumerate
all possible clauses in \{A\}, all possible clauses in \{B\}, and
all possible clauses in \{C\} (and those in \{not A\}, etc.) as well
as the possible logic operators, <and/or/xor>, between the clauses.
This approach is computationally intensive because the number of all
possible combinations grows exponentially in the total number of clauses
$\ell$. For $p$ covariates discretized into $q$ candidate clauses
this leads to $3^{\ell-1}(2pq)^{\ell}$ different combinations. To
put that number into context, fifty variables discretized into three
candidate clauses each has $243$ million different combinations to
search over for a three-clause sentence.

Practically, I can use the brute force algorithm to enumerate all
possible targeting policies when $\ell\leq2$ and the dataset itself
is not too large. I do so in Appendix Section \ref{sec:Bounding-the-greedy-algorithm}
to provide a comparison to the greedy algorithm's solution for the
empirical application.

\subsection{Greedy algorithm\label{subsec:Greedy-algorithm}}

Since the brute force algorithm is not computationally tractable in
many scenarios, I propose a greedy version of the algorithm.\footnote{Greedy algorithms are commonly used in the marketing literature \citep{Lilien1992}
and also used in estimating decisions trees in the statistics literature
\citep{Breiman1984}.} The algorithm is computationally feasible but may only find a local
optimal solution that generates lower profits than that of the globally
optimal solution. The greedy algorithm's profits can be viewed as
a lower bound of the profits from the globally optimal comprehensible
policy; I provide worst case profit bounds for the greedy algorithm
in Appendix Section \ref{subsec:Worst-case-bounds}. 

\begin{algorithm}
\caption{Greedy algorithm \label{alg:Greedy-algorithm}}

\textbf{Setup}: Number of clauses $\ell$:
\begin{enumerate}
\item Discretize the $p$ covariates into $q$ pieces to get the clauses
\item Find the single best clause that maximizes profits
\item For $l\in\{2,\ldots\ell\}$:
\begin{enumerate}
\item Iterate all clause and logic operator combinations while holding the
$l-1$ clauses and logic operators fixed
\item Choose the combination that maximizes the profits
\end{enumerate}
\end{enumerate}
\end{algorithm}

To outline the greedy algorithm (Algorithm \ref{alg:Greedy-algorithm}),
consider the case of finding a $\ell=3$ clause sentence that has
the structure: 
\[
\text{\ensuremath{\text{Target if }\underset{\ensuremath{(1)}}{\underbrace{\text{\{A\}}}}\underset{(2)}{\underbrace{\text{<and/or/xor> \{B\}}}}~\underset{(3)}{\underbrace{\text{<and/or/xor> \{C\}}}}}}
\]
where $\text{\{A\}, \{B\}, \{C\}}$ are the clauses in the sentence
and <and/or/xor> are the logic operators. The greedy algorithm breaks
up the combinatorially difficult problem by solving it piece by piece.
In the example, the greedy algorithm would first find the best single
clause sentence or optimize over the possible clauses \{A\} (the first
piece). Then, it would hold the one clause targeting rule (the solution
to \{A\}) fixed and then find the best logic operator and \{B\} combination
(the second piece). Lastly, it would hold the two clause targeting
rule (the solution to \{A\} <and/or/xor> \{B\}) fixed and find the
best logic operator and \{C\} combination (the third piece).

The greedy algorithm evaluates a smaller set of combinations of comprehensible
policies since it does not enumerate over all possible combinations
of clauses and logic operators. As a result, it searches over $6(\ell-1)pq\ell$
combinations for a comprehensible policy with $\ell$ clauses and
$p$ covariates discretized in $q$ candidate clauses. For fifty variables
discretized into three candidate clauses, there are approximately
$5.4$ thousand combinations to search over for a three-clause sentence,
which is many orders of magnitude smaller than the brute force approach's
$243$ million combinations. 

In the empirical application, I use the greedy algorithm to solve
for the optimal comprehensible policy. Since the combinatorial space
is dramatically reduced, searching for the optimal comprehensible
policy with the greedy algorithm for $\ell=10$ takes around a minute
while using the brute force algorithm for only $\ell=3$ has an estimated
run time of over three weeks. Both algorithms were implemented using
the R package for \texttt{torch} as the backend \citep{Falbel2023}.

\subsection{Inference for optimal comprehensible policies\label{subsec:Comp-Inference}}

To conduct inference around the optimal comprehensible policy, whether
it is found though the brute force algorithm or the greedy algorithm,
I can adopt results from \citet{Kitagawa2018}.\footnote{To complete their setup, I would need to further assume the outcome
variable ($Y)$ is bounded in addition to the standard assumptions
of Assumptions \ref{assu:Unconfoundedness}, \ref{assu:Overlap},
and \ref{assu:SUTVA}.} I am using the empirical welfare maximization framework, albeit with
a more specific function class $\mathcal{F}_{\text{comp}}^{\ell}$
and a fixed $\ell$. I now summarize how I implement their theoretical
results in this framework and further theoretical details can be found
in their paper.

\citet{Kitagawa2018} provide a minimax optimal rates for policy learning
via empirical welfare maximization. They study expected welfare regret
of a candidate policy function's welfare to the optimal policy function's
welfare. Their minimax rates around expected regret provide worst
case guarantees for finding the optimal policy function and the rates
scale at $K\sqrt{VC(d)/n}$ where $VC(d)$ represents the Vapnik-Chervonenkis
(VC) dimension of the policy function, $K$ is a constant, and $n$
is the number of observations. The VC dimension of the class of policy
functions inherently needs to be finite. To adapt the results, I define
the class of policy functions to be the class of comprehensible policies
$\mathcal{F}_{\text{comp}}^{\ell}$ with $\ell$ fixed.\footnote{The number of clauses $\ell$ need to be fixed and finite in order
for the VC dimension of the comprehensible policy to be finite (Lemma
\ref{lem:Comprehensible-Sentence-Finite-VC-Dim} in Appendix Section
\ref{sec:Decision-Trees-and-Sentences}).} 

In this setting, I only need rates for statistical inference, which
are a different objective than minimax rates. Confidence intervals
can be attained around the estimated policy function using the empirical
process bootstrap outlined in Algorithm B.1 in Appendix B of \citet{Kitagawa2018}.
I follow this procedure to attain inference around the optimal comprehensible
policy in the empirical application.

\section{Projecting down the black box\label{sec:Ex-post-policy}}

In this section, I show how comprehensible policies can be formed
by projecting down a black box policy and demonstrate that doing so
will yield less profitable comprehensible policies than finding them
directly from the data. I call this projection down procedure the
\emph{ex post }approach\emph{ }and call the procedure for finding
the policy directly from the data (Section \ref{sec:Optimal-comprehensible-policy})
the \emph{direct }approach. 

I first provide an analytical justification for why the \emph{direct}
approach will generate more profitable comprehensible policies than
the \emph{ex post} approach in Section \ref{subsec:Projecting-down-black-box-not-profit-max}.
I then show how I can attain inference around the projected down optimal
comprehensible policy by recentering the empirical process results
from \citet{Kitagawa2018} in Section \ref{subsec:Inference-for-projection-down}.

The literature in explainable AI (XAI) studies a local approximation
of a black box algorithm where they project down the black box to
a simpler, more explainable model \citep{Biran2017a,Miller2019a}.
However, I will show that using the projection down procedure to form
optimal comprehensible targeting policies leads to less profitable
comprehensible policies than those formed directly from the data (as
in Section \ref{sec:Optimal-comprehensible-policy}).

To project down the black box, I first consider a profit loss function
from two candidate targeting policies $d(x)$, $d'(x)$. I consider
profits as the outcome of interest because they are a direct measure
of producer surplus while other metrics such as AUC or classification
accuracy to the black box method do not have a direct economic interpretation. 

I choose the absolute profit difference as the loss function between
two policies $d(x)$ and $d'(x)$, expressed as 
\begin{align}
\mathcal{L}(d,d') & =|\hat{\Pi}(d)-\hat{\Pi}(d')|\nonumber \\
 & =\sum_{i=1}^{n}\underbrace{\mathbf{1}\{d(x_{i})\neq d'(x_{i})\}}_{\text{Classification loss}}\underbrace{\left|\frac{W_{i}}{e(x_{i})}\pi_{i}(1)-\frac{1-W_{i}}{1-e(x_{i})}\pi_{i}(0)\right|}_{\text{Weight}}.\label{eq:IPWE-Profit-Difference}
\end{align}
The first term, $\mathbf{1}\{d(x_{i})\neq d'(x_{i})\}$, can be interpreted
as the classification loss.\footnote{Appendix Section \ref{sec:Profit-loss-derivation} provides the derivation
of the loss function.} The second term, $|\frac{W_{i}}{e(x_{i})}\pi_{i}(1)-\frac{1-W_{i}}{1-e(x_{i})}\pi_{i}(0)|$,
can be interpreted as the classification weight for customer $i$.
This loss function is similar to the outcome weighted learning setup
from \citet{Zhao2012} and the weighted classifier setup in \citet{Zhang2012}.
Intuitively, the loss is nonzero for an individual if the two policies
differ (the classification loss) and the difference is scaled by the
absolute profit difference from disagreeing for that individual. Further,
the expected weight for the observation is $E\left[\frac{W_{i}}{e}\pi_{i}(1)-\frac{1-W_{i}}{1-e}\pi_{i}(0)\right]=E[\pi_{i}(1)-\pi_{i}(0)]$
or the expected individual level profit difference.

Using the absolute profit difference loss function in the XAI framework,
I aim to find the explainable policy $\xi(x)$ that solves the equation
\begin{equation}
\xi(x)=\underset{d'}{\arg\min}~\mathcal{L}(d,d')+\Omega(d')\label{eq:XAI-Formula}
\end{equation}
 where $\mathcal{L}(d,d')$ is the difference in the profits of the
two policy functions $d$, $d'$ as in Equation \ref{eq:IPWE-Profit-Difference}
and $\Omega(d')$ represents the complexity of explainable policy
$d'$. The generalized loss function in Equation \ref{eq:XAI-Formula}
trades off between minimizing the loss between $d$ and $d'$ and
minimizing the complexity of $d'$.

Equation \ref{eq:XAI-Formula} captures the general framework used
by the XAI literature where an explainable, transparent box model
$d'$ is used to approximate the black box model $d$ \citep{Biran2017a,Miller2019a}.
Methods in this domain include LIME \citep{Ribeiro2016} and SHAP
\citep{NIPS2017_8a20a862}, which both provide a local linear decomposition
of the black box to its covariates. 

I take a step further and directly embed explainability as a constraint
in the problem. Instead of considering a trade-off between model complexity
and the performance of the explainable policy, I choose a level of
complexity and then find the best performing explainable policy. Consider
an explainable class of policies $D(\Omega_{l})$ that all have the
same level of complexity $\Omega_{l}$, which is indexed by some constant
$l$. I now want to find the best performing explainable policy $d'\in D(\Omega_{l})$
and solve the equation
\begin{equation}
\xi_{D'}(x)=\underset{d'\in D(\Omega_{l})}{\arg\min}~\mathcal{L}(d,d')\label{eq:ex-post-Loss}
\end{equation}
where I embedded the level explainability as a direct constraint in
the optimization problem. 

I let $d$ be the estimated black box policy function $d_{BB}^{*}(x)$.
For the explainable policy, I use the class of comprehensible policies
constructed in Section \ref{sec:Comprehensible-class} and the complexity
parameter $l$ can just be the number of clauses $\ell$ in the sentence.
To unify the notation, I define $\mathcal{F}_{\text{comp }}^{\ell}=D(\Omega_{l})$
as the class of comprehensible policies of length $\ell$.

\subsection{Projecting down the black box is not profit-maximizing \label{subsec:Projecting-down-black-box-not-profit-max}}

I now show that projecting down the black box policy to find the comprehensible
policy following Equation \ref{eq:ex-post-Loss} leads to a less profitable
comprehensible policy than directly optimizing the comprehensible
policy by maximizing profits (Section \ref{sec:Optimal-comprehensible-policy}).
I first compare the two objectives used in the two approaches,
\begin{equation}
\min_{d'\in\mathcal{F}_{\text{comp}}}\left|\max_{d\in\mathcal{F}_{BB}}\hat{\Pi}(d)-\hat{\Pi}(d')\right|~\text{ vs. }\max_{d'\in\mathcal{F}_{\text{comp}}}\hat{\Pi}(d').\label{eq:Objective-comparisons}
\end{equation}
The left-hand side represents the profit level from forming a projected
down optimal comprehensible policy using Equation \ref{eq:ex-post-Loss}
and the right-hand side represents the profit level from forming the
optimal comprehensible policy directly from the data. When finding
the optimal comprehensible policy directly, $d_{\text{comp }}^{*}(x)=\arg\max_{d'\in\mathcal{F}_{\text{comp}}}\hat{\Pi}(d')$,
the property of the maximization operator says the \emph{direct} approach
will find the maximum profits for the given comprehensible policy
class. Thus, the comprehensible policy from the \emph{ex post} approach
will generate weakly less profits than the direct approach. 

The intuition behind the result is lies in the structure of the loss
functions for the two approaches. Because the ex post approach minimizes
the loss between the comprehensible policy and the black box, in cases
where the black box does not classify customers well, the ex post
approach will also not do well. In contrast, the direct approach learns
the policy straight from the data as it does not depend on the black
box policy's results. In fact, if the black box does better than the
comprehensible policy in generating a more profitable policy for all
individuals in the data, then the two approaches will be the same.

I formalize the last statement to show that the two objectives in
Equation \ref{eq:Objective-comparisons} will be equal if the black
box policy outperforms the comprehensible policy for all customers
with the following assumption. I denote the individual-level profits
from targeting policy $d(x_{i})$ as $\pi(d(x_{i}))$.
\begin{assumption}
\label{assu:DNN-more-profitable} The profits generated from the black
box are weakly greater than that of the comprehensible policy, so
$\pi(d_{BB}(x_{i}))\geq\pi(d(x_{i})),\forall d\in\mathcal{F_{\text{comp }}^{\ell}}$. 
\end{assumption}

Under Assumption \ref{assu:DNN-more-profitable}, 
\begin{align*}
\min_{d'\in\mathcal{F}_{\text{comp}}^{\ell}}\left|\max_{d\in\mathcal{F}_{BB}}\hat{\Pi}(d)-\hat{\Pi}(d')\right| & =\min_{d'\in\mathcal{F}_{\text{comp}}^{\ell}}\left(\max_{d\in\mathcal{F}_{BB}}\hat{\Pi}(d)-\hat{\Pi}(d')\right)\\
 & =\min_{d'\in\mathcal{F}_{\text{comp}}^{\ell}}-\hat{\Pi}(d')\\
 & =\max_{d'\in\mathcal{F}_{\text{comp}}^{\ell}}\hat{\Pi}(d'),
\end{align*}
where I used the assumption in the first line to remove the absolute
value as $\max_{d\in\mathcal{F}_{BB}}\hat{\Pi}(d)=\hat{\Pi}(d_{BB}^{*}(x))\geq\hat{\Pi}(d')$
for $d'\in\mathcal{F}_{\text{comp}}^{\ell}$ and I used that $d'$
does not show up in the first term to get to the second line. 

While black box algorithms converge to the true policy function asymptotically
(satisfying Assumption \ref{assu:DNN-more-profitable}), overfitting
can degrade their performance in finite samples. Projecting down an
overfit black box model produces comprehensible rules inferior to
those found by direct optimization.

I interpret these results as cautionary guidance for marketing managers.
The standard XAI approach\textemdash projecting down the black box
to a smaller explanatory model\textemdash prioritizes approximating
the black box over maximizing firm objectives. Firms seeking to implement
comprehensible policies should directly optimize profits rather than
inheriting deficiencies of intermediate black box models.

\subsection{Inference for projected down optimal comprehensible policies \label{subsec:Inference-for-projection-down}}

To conduct inference around the \emph{ex post} approach, I recenter
the empirical process results from \citet{Kitagawa2018} for the loss
function in Equation \ref{eq:ex-post-Loss}. By showing their results
apply to my framework, I can conduct inference for the \emph{ex post}
comprehensible policy.

I first define the notation for the theorem and suppress the dependence
of the comprehensible policy targeting rule $d$ on the length of
the sentence $\ell$ for notational simplicity. These results holds
for a finite number of clauses $\ell$, implying the comprehensible
policy has finite VC dimension (Appendix Lemma \ref{lem:Comprehensible-Sentence-Finite-VC-Dim}).
I first define $\pi(d)$ to be the individual-level profits from targeting
policy $d$, and then define
\begin{align*}
\check{d}^{*}(x) & =\arg\max_{d\in\mathcal{F}_{\text{comp}}}E_{P}\left[\left|\pi(d_{BB}^{*})-\pi(d)\right|\right]\\
\hat{d}_{ex}(x) & =\arg\min_{d\in\mathcal{F}_{\text{comp}}}E_{n}\left[\left|\pi(d_{BB}^{*})-\pi(d)\right|\right],
\end{align*}
where the first line has the expectation taken over the population
distribution $P$ and the second line has the expectation taken over
the sample analog. The second line produces the \emph{ex post} optimal
comprehensible policy $\hat{d}_{ex}(x)$ and is equivalent to finding
the optimal comprehensible policy by solving the sample analog of
Equation \ref{eq:ex-post-Loss}.\footnote{The formal definition of $\pi(d)$ is in Equation \ref{eq:Individual-Profit-Estimator}
in Appendix Section \ref{sec:Profit-loss-derivation}. }

I impose sample splitting: I estimate $d_{BB}^{*}(x)$ for the optimal
black box policy in one data sample, find the optimal comprehensible
policy $\hat{d}_{ex}(x)$ in another data sample, and conduct inference
in the last data sample. I further define $\Gamma(d)=E_{P}\left[\left|\pi(d_{BB}^{*})-\pi(d)\right|\right]$
to be the profit loss for the targeting rule $d$ evaluated in the
population and $\Gamma_{n}(d)=E_{n}\left[\left|\pi(d_{BB})-\pi(d)\right|\right]$
to be the profit loss for targeting rule $d$ in the sample.
\begin{thm}
\label{thm:Convergence-rate-of-ex-post-comprehensible-policies}(Uniform
convergence rate of ex post comprehensible policies) Under Assumptions
\ref{assu:Unconfoundedness}, \ref{assu:Overlap}, and \ref{assu:SUTVA}
and the assumption that the outcome variable ($Y$) is bounded, for
a hypothesis space of comprehensible policies with $\ell$ clauses
($\mathcal{H}=\mathcal{F}_{\text{comp}}^{\ell}$) that has bounded
VC dimension ($VC(\mathcal{H})<V<\infty$), 
\[
\sup_{P}E_{P}\left[\Gamma(\check{d}^{*})-\Gamma(\hat{d}_{ex})\right]\leq C\sqrt{\frac{V}{n}}=O_{p}\left(1/\sqrt{n}\right).
\]
\end{thm}

I provide the proof of Theorem \ref{thm:Convergence-rate-of-ex-post-comprehensible-policies}
in Appendix Section \ref{sec:Proofs-for-XAI}. The technical implication
is that the rate for the difference in the profit loss scales at $\sqrt{VC(d_{ex})/n}$
for \emph{ex post} comprehensible policy $d_{ex}$. This is an upper
bound and the lower bound results from \citet{Kitagawa2018} can be
similarly recentered to provide minimax rates for the profit differences.
Inference can be attained around this \emph{ex post }approach by using
the empirical process bootstrap.

\section{Empirical application\label{sec:Empirical-application}}

In this section, I provide an application to promotions management
for a durable goods retailer as a proof of concept of the methodological
framework. I first find the best performing black box to establish
the optimal black box policy benchmark $d_{\text{BB}}^{*}(x)$. Then,
I find the optimal comprehensible policy $d_{\text{comp}}^{*}(x)$
and compare the targeting differences between the two targeting policies.
I denote the profit difference of the best-performing black box to
the comprehensible policy as the cost of explanation when implementing
the comprehensible policy.\footnote{This framework is more general than shown in this application. A similar
analysis can be performed for any black box policy class and any comprehensible
policy class.} I then show that finding the optimal comprehensible policy directly
produces a more profitable targeting policy than finding it by projecting
down the black box targeting policy. 

I use the second ISMS Durable Goods dataset from \citet{Ni2012} that
contains a price promotion randomized control trial (RCT) for a durable-goods
store in 2003. The items available are mainly electronics and they
encompass a range of products from small ticket to large ticket items.\footnote{To provide a concrete example, a sample small ticket item is something
like a drip coffee maker and a sample large ticket item is something
like a refrigerator.} The price promotion is a \$10-off coupon that is valid on the next
in-store purchase.

The RCT contains 176,961 customers, and the treated customers were
sent a promotion with probability $50\%$. The control group was not
mailed any promotion. The dataset contains approximately $150$ recency,
frequency, and monetary (RFM) covariates that describe the customers'
past behavior with the firm. The outcome of interest is sales during
December 2003 (the promotional period) and the price promotion was
mailed to the customers before December 2003. 

To check for evidence of the covariate balance and the overlap assumption,
I run a logistic regression to check if I can statistically predict
the treatment variable of getting the promotion with the RFM data.
I find that only the intercept value is statistically significant
in the regression. Appendix Figure \ref{fig:Prop-Score-Overlap} shows
the density of estimated propensity scores for the treated and the
not treated groups from the logistic regression. The two densities
essentially fully overlap which suggests that overlap and covariate
balance hold in the data. Further checks for covariate balance can
be found in \citet{Ni2012}.

These results suggest the RCT was correctly implemented so the unconfoundedness
and overlap assumptions should hold in the data. I further make the
assumptions that customers used the coupon on their next possible
purchase, there's no gaming of coupons, and there are no spillover
effects from the mailed promotions to satisfy the stable unit treatment
value assumption (SUTVA). 

With the three standard assumptions satisfied, I can estimate the
average treatment effect (ATE) of the price promotion on December
2003 and find the ATE to be $2.68$ with a $95\%$ confidence interval
of $(1.63,3.73)$. This result suggests that there is a statistically
significant effect of the price promotion on December sales.

From the sales data, I also see that only $3.6\%$ of all customers
purchase during December. Further, conditional on purchase, the median
spend size is $\$149.99$. There are a handful of individuals in the
data who spend over five thousand in the store. Since I did not collect
the data, I am not sure if these are outliers or errors in the data.
In my subsequent analysis, I drop those that spend more than $\$800$
at the store (the top $0.29\%$ of spenders) in the data.\footnote{I further motivate this data cleaning procedure by assuming that if
a customer decides to spend a few thousand at the store than they
are less influenced by the \$10-dollar off promotion to make the purchase
because it is effectively a smaller percentage off the base price.
}

I impose profit margins $m=45\%$ and cost of mailing $c=37\cent$
to complete the setup. The latter is the price of mailing a letter
at USPS during the time period. Alternative numbers for the profit
margin can be used in the framework and they can even be set to vary
by customer covariates. 

I first randomly split the data 80/20, estimate the models with 80\%
of the data, and evaluate different methods' targeting policies out
of sample using 20\% of the data.\footnote{I denote evaluating models in the $80\%$ training data as in sample
evaluation and in the 20\% validation data as out of sample evaluation.} I use the Policy DNN, Causal DNN, Causal Forest, and Lasso black
box methods and implementation details are provided in Appendix Section
\ref{sec:Implementation-Details}.

Causal Forests and Causal DNNs are commonly used as a state-of-the-art
procedure in applied economics and marketing literatures \citep{Wager2018,Farrell2021}.
The Causal Forest bootstrap aggregates the Causal Trees from \citet{Athey2016}.
Lasso is a popular machine learning method when linearity of the baseline
model and the heterogeneous treatment model is assumed \citep{Hastie2015,Taddy2019}.
The Causal Forest, Causal DNN, and Lasso first estimate $\hat{\beta}(x)$
which is then plugged in optimal targeting policy function. In contract,
Policy DNN directly learns the policy function $\hat{d}(x)$ and represents
the policy function as deep neural network \citep{Zhang2024}. These
four methods provide benchmark black box models used in the literature.

With the black box targeting policies, I can evaluate the expected
profits under the targeting policy using Equation \ref{eq:IPWE-Profits}.
Table \ref{tab:Black-Box-Profits} provides the out of sample individual
expected profits from these black box models as well as from a blanket
targeting policy where everyone is sent the promotion. I see that
the Policy DNN ($\$3.03$) does better than both the Causal Forest
($\$2.78$) and Lasso ($\$2.70$) procedures in generating profits.
All black box methods perform better than the blanket targeting policy
($\$2.42$). 

I interpret the profit gap between the Policy DNN and the Causal DNN
as the difference from the policy learning approach where the optimal
targeting policy is learned directly from the data to the standard
approach that first learns the heterogeneous treatment effects and
then plugs them into the optimal treatment rule. Less information
is needed to learn the optimal policy function directly than to learn
the heterogeneous treatment effects from the data. By focusing on
only learning what is needed for the targeting policy, the procedure
attains better performance in the dataset. 

I then interpret the small profit gap between the Causal Forest and
the Lasso as reflecting that the heterogeneous treatment effects can
be well approximated by a sparse linear functional form. Lastly,
the gap between the Lasso and blanket mailing represents the difference
between doing any personalization using a black box algorithm and
doing no personalization. This profit gap is as large as the gap between
the Policy DNN and the Causal Forest. 

For the remainder of the analysis, I use the Policy DNN as the black
box benchmark, as it achieves the highest profits among all methods
tested. This choice reflects realistic firm behavior: in practice,
firms would deploy their best-performing unconstrained model. The
framework generalizes to any black box approach\textemdash Policy
DNN simply represents the strongest baseline for quantifying the cost
of comprehensibility constraints in this application.

I also examine Policy Tree, which is an interpretable model that projects
Causal Forest estimates onto a decision tree \citep{Athey2021}. Table
\ref{tab:Black-Box-Profits} reports profits for a depth-2 Policy
Tree: it underperforms the Causal Forest but outperforms the Lasso.
This ranking suggests that second-order interactions capture some
valuable heterogeneity in targeting decisions.\footnote{A larger Policy Tree is computationally prohibitive due its global
search over all possible trees.} 

I form the optimal comprehensible targeting policy following Section
\ref{sec:Comprehensible-class} and Section \ref{sec:Optimal-comprehensible-policy}.\footnote{The brute force solution for a two clause targeting policy is in Appendix
Section \ref{sec:Implementation-Details}.} I denote the optimal comprehensible policy as $d_{\text{comp }}^{*}(x)$,
and the three-clause optimal comprehensible policy using the greedy
optimization algorithm is:

\begin{align*}
\text{Target customer if she:}\\
1. & \text{ has bought \emph{high} amount of items during Christmas over the last two years }\\
 & and\\
2. & \text{ did \textbf{not} have \emph{high} spending during Spring over the last two years}\\
 & or\\
3. & \text{ has \emph{low} spending during last years holiday mailer promotional period}.
\end{align*}
The customers that are not described by the targeting policy are not
targeted. The construction of the clauses with \emph{low} and \emph{high}
descriptors from the RFM dataset follow the discussion in Section
\ref{subsec:Generating-clauses}. 

I interpret this targeting sentence as largely targeting two major
segments of customers. The first group is described by the first two
clauses. These are customers who buy a lot during Christmas but not
a lot in the Spring, or people who focus their spending during the
holiday period. Since the outcome of interest is December sales, these
are individuals who spend a lot during the target period and may be
more price sensitive.

The second group are those who are on the RFM customer list but have
not spent a lot during the promotional December period the prior year.
I consider this customer group as those who spent at the store in
years before but then either forgot about the store or went to another
store the previous year. The optimal comprehensible policy suggests
retargeting these customers and incentivize them to come back to the
store. Customers in either of these two groups will be targeted by
the optimal three-clause comprehensible policy.

\subsection{Do the targeting policies differ?\label{subsec:Targeting-Policy-Difference}}

I first compare the targeting differences between Policy DNN $d_{BB}^{*}(x)$
and the optimal comprehensible policy $d_{\text{comp}}^{*}(x)$, which
I call the \emph{direct} method in the figures and tables.\footnote{Section \ref{sec:Ex-post-policy} makes the distinction between the
\emph{direct} and \emph{ex post} methods for finding the optimal comprehensible
policy.} Figure \ref{fig:Targeting-Percentage} shows the targeting percentage
of the customers for $\ell\in\{1,\ldots10\}$ number of clauses. Policy
DNN targets $18.1\%$ of customers. I see that with one clause, the
optimal comprehensible policy targets more than Policy DNN. However,
with three clauses, it gets to the closest in targeting percentage
to Policy DNN. Adding more clauses in this setting appears to reduce
the overall targeting percentage of the comprehensible policy.

I provide a visual demonstration of the targeting policy differences
with a two-clause optimal policy in Figure \ref{fig:Two-Clause-Diagram}.\footnote{The two-clause comprehensible policy is the just the first two clauses
of three clause targeting policy from before, ``Target if customer
has high amount of items during Christmas over the last two years
and did not have high Spring spending over the last two years.''} The axes represent two RFM covariates in the dataset: \emph{spring
sales over the last 24 months} and \emph{items bought during Christmas
over the last 24 months}. The jittered grey circular points represent
the customers in the raw dataset. The jittered green triangular points
are customers that Policy DNN targets. Since Policy DNN learns a higher
order representation of the data, it is not clear what its targeting
rule is when visualized on these two dimensions. In contrast, the
two-clause optimal comprehensible policy is represented by the pink
rectangle, and everyone covered by the rectangle will targeted by
the comprehensible policy. 

I now consider the three-clause optimal comprehensible policy and
provide the confusion matrix in Table \ref{tab:Targeting-policy-differences}.
I see that the three-clause policy targets 18.3\% of the customer
base and there is a 75.9\% overlap between Policy DNN and the three-clause
targeting policy. More specifically, the two policies agree in targeting
6.3\% of the customers and not targeting 69.6\% of the customers.
The three-clause comprehensible policy targets 11.8\% of customers
who are not targeted by Policy DNN and Policy DNN targets 12.3\% of
the customers not targeted by the three-clause comprehensible policy.
Since the overtargeting and undertargeting differences are relatively
balanced when comparing the comprehensible policy to Policy DNN, it
seems that the comprehensible policy is capturing similar variation
as Policy DNN but cannot personalize as finely due to its comprehensibility
constraint. 

\subsection{Cost of explanation\label{subsec:Targeting-Policy-Profit-Difference}}

I now quantify the profit differences from Policy DNN and the optimal
comprehensible policy and denote this gap as the \emph{cost of explanation}.
In Figure \ref{fig:Individual-Expected-Profits}, I visualize the
expected individual profits for the optimal comprehensible policy
by the number of clauses $\ell\in\{1,\ldots,10\}$. I focus on the
\emph{direct} method for now and see that in sample as the number
of clauses increases, the comprehensible policy does better. This
result captures the fact that with more clauses the comprehensible
policy can make more partitions of customers and better personalize
the targeting policy. 

Out of sample, the profits increase with more clauses but decrease
slightly after eight clauses. These results suggests that with nine
or ten clauses, the direct method can still overfit to the training
data. I also see that out of sample, the gap between the optimal comprehensible
policy and Policy DNN is smaller because Policy DNN is likely overfitting
in sample. After three clauses, the comprehensible policy's expected
profits is not statistically significantly different from that of
Policy DNN. 

I evaluate the cost of explanation, or the profit difference of the
two policies, for the three-clause optimal comprehensible targeting
policy in Table \ref{tab:Cost-of-Explanation}. Per person, the optimal
comprehensible policy generates $\$2.80$ in expected profits and
the out of sample cost of explanation is $23$ cents. This implies
that implementing the three-clause optimal comprehensible policy policy
instead of the Policy DNN black box policy will lead to a $23$ cents
loss in expected profits per person. This is a $7.5\%$ loss compared
to the Policy DNN profits and the comprehensible policy provides a
$16\%$ gain in profits compared to a blanket mailing policy. 

Further, the difference between Policy DNN and the blanket mailing
is $61$ cents. The optimal comprehensible policy provides a $38$
cents gain over blanket targeting, or recouping $(60-23)/(60)=62\%$
of the expected losses from implementing the blanket mailing policy.
Thus, the firm does notably better by implementing the three-clause
comprehensible policy than a blanket mailing policy.

Notably, the optimal comprehensible targeting policy (\$2.80) outperforms
both the Causal Forest (\$2.78) and Lasso (\$2.70) in out-of-sample
profits (Table \ref{tab:Black-Box-Profits}). While these differences
are not statistically significant, they suggest potential advantages
of direct policy learning over two-step approaches, even for simple
comprehensible policies. Similarly, the Policy DNN outperforms the
Causal DNN.

I interpret these results as reflecting the limited treatment heterogeneity
in this application. Comprehensible policies perform well when the
true optimal policy is relatively simple. With greater heterogeneity
or larger datasets enabling complex pattern detection, flexible black
box models should dominate.

\subsection{Projecting down the black box}

I apply the\emph{ ex post} approach from Section \ref{sec:Ex-post-policy}
to project the black box model down to a comprehensible policy. I
use the same 80/20 data split to evaluate the methods but split the
training sample again.\footnote{Policy DNN is trained on $40\%$ of the data, the DNN is projected
down to form the \emph{ex post }comprehensible policy in $40\%$ of
the data, and inference is conducted in the last $20\%$ of the data.}

Figure \ref{fig:Individual-Expected-Profits} plots the individual
expected profits for the \emph{direct} and \emph{ex post} methods.
I see that the \emph{direct} method outperforms the\emph{ ex post}
method out of sample. In expectation, the direct method generates
$\$2.80$ per person while the ex post procedure only generates $\$2.59$
per person. The \emph{ex post} procedure's targeting policy is, 
\begin{align*}
\text{Target customer if she:}\\
1. & \text{ has bought \emph{high} amount of items during Christmas over the last two years }\\
 & or\\
2. & \text{ has \emph{high} amount of back to school items over the last two years}\\
 & or\\
3. & \text{ has \emph{no} purchased item during the Expo promotion two years ago}.
\end{align*}
Comparing this targeting policy that of the \emph{direct} approach,
I see that their first clauses are identical, but their second and
third clauses as well as their logic operators are different. The
differences imply that the two approaches are capturing different
partitions of the customer base to target.

The \emph{ex post} approach targets three segment of customers:
It targets customers who buy a lot during Christmas, buy a lot of
back to school items, or did not purchase during the Expo promotion
two years ago. The first two segments are those who consistently exhibit
significant spending patterns at the end of the year and the last
segment are newer customers or those not responsive to promotions
during a consumer electronics exhibition. 

Table \ref{tab:Cost-of-Explanation} shows the cost of explanation
for the two methods with three clauses. I see that the cost of explanation
for the \emph{ex post} procedure (53 cents) is higher than that of
the the \emph{direct} procedure out of sample (23 cents) out of sample.
These results match the profit differences visualized in Figure \ref{fig:Individual-Expected-Profits}
and suggest the \emph{direct }approach generates more profitable comprehensible
policies. 

Further, a two-clause optimal comprehensible policy learned directly
generates $\$2.76$ per person, slightly outperforming the depth-2
Policy Tree ($\$2.72$). Despite the Policy Tree's more flexible functional
form, direct optimization yields higher profits. Since the Policy
Tree projects down from the Causal Forest, this result reinforces
that comprehensible policies should be optimized directly rather than
derived from black box models.

Overall, this application empirically verifies the analytical results
from Section \ref{subsec:Projecting-down-black-box-not-profit-max}.
Optimal comprehensible policies should be found directly from the
data (Section \ref{sec:Optimal-comprehensible-policy}) rather than
found by projecting down from the black box policy (Section \ref{sec:Ex-post-policy}).

\section{Discussion\label{sec:Discussion}}

In this section, I study how firm managers can use the the proposed
framework to analyze their decision to stay with a black box algorithm
or to move to a comprehensible policy when forming targeting policies
for their marketing mix. Circling back to the framework overview in
Figure \ref{fig:Methodological-overview}, I revisit the trade-off
between profits and comprehensibility. Managers compare the complete
producer surplus generated by the two policies, or
\[
\underbrace{\Pi_{BB}-R_{BB}+B_{BB}}_{\text{Black Box}}\text{ vs. ~\ensuremath{\underbrace{\Pi_{comp}-R_{comp}+B_{comp}}_{\text{Comprehensible Policy}}}},
\]
to make the decision. I first highlight three components of the complete
producer surplus to study the manager's problem of whether to stay
with the black box or move to a comprehensible policy. 

The first term ($\Pi_{BB},\Pi_{comp}$) represents the short term
profit loss from moving away from the black box to the optimal comprehensible
policy. The proposed framework constructs the black box and optimal
comprehensible policy benchmarks, and the analysis in Section \ref{sec:Empirical-application}
quantifies the short term profit loss, or the cost of explanation,
in the empirical example.

The second term ($R_{BB},R_{comp}$) represents the regulatory penalty
that the firm faces while implementing its chosen targeting algorithm.
If enforced, right-to-explanation laws will penalize black boxes and
but not comprehensible targeting policies. Thus, the expected regulatory
penalty will be higher for the black box policy ($R_{BB}>R_{comp}$).

The third term ($B_{BB},B_{comp}$) represents the long-term effects
of offering a comprehensible policy. Consumers can firms can benefit
from comprehension as it can build better brand equity for the firm
and makes implementing the targeting policy easier by its representatives.
If comprehension leads to long-term benefits or costs, then it should
be considered by firm managers when making the decision.

In Section \ref{subsec:GDPR-Calibration}, I focus on the first two
terms that balance profitability and with the expected regulation
penalty ($\Pi_{BB}-R_{BB}$ vs. $\Pi_{comp}-R_{comp}$). I leverage
the proposed framework to study the effect of right-to-explanation
legislation on firm's profits as the firm moves to an optimal comprehensible
policy. This calibration exercise locally quantifies the economic
impact that right-to-explanation legislation imposes on firms if enforced. 

In Section \ref{subsec:Long-term-effects-of-Comprehensibility}, I
then discuss the possible long-term effects of comprehensible targeting
policies. Even though I do not have the data to study the benefits
of offering a comprehensible policy in my empirical application, I
outline potential factors that firms should consider when forming
their decision.

\subsection{Effect of GDPR's \textquotedblleft right to explanation\textquotedblright{}
on firms \label{subsec:GDPR-Calibration}}

The EU's GDPR legislation states customers have a ``right to explanation''.
This clause would require firms to employ human representatives to
provide both ``an explanation'' and ``meaningful information about
the logic involved'' behind any decision made by an algorithm \citep{EuropeanCommission2016}.\footnote{Specifically, the GDPR legislation states:

\textquotedblleft {[}The data subject should have{]} the right to
obtain human intervention, to express his or her point of view, to
obtain an explanation of the decision reached\textquotedblright{} 

\textquotedblleft {[}and given{]} access to meaningful information
about the logic involved\textquotedblright{} 

\hfill{}\citep{EuropeanCommission2016}} If the right-to-explanation clause is enforced, the penalties for
violating GDPR are the larger of $4\%$ of global revenues or $20$
million Euros. Although how exactly the right-to-explanation laws
apply to firms is a subject of ongoing legal debate \citep{Wachter2016},
it is prudent for forward-looking firms to consider their potential
implications. With my framework, I evaluate their impact on profits
as a firm transitions from a black box targeting policy to an optimal
comprehensible targeting policy to comply with the legislation.

From the empirical application, I showed the cost of explanation ($\Pi_{BB}-\Pi_{comp}$)
was $23$ cents per person for the three-clause optimal comprehensible
policy. For a customer basis of $10$ million, this implies $2.3$
million dollars of lost profits due to moving away from the black
box targeting policy.

GDPR litigation and enforcement has been publicly focused on multinational
technology firms, but all firms under EU jurisdiction must abide by
the law. To put the lost profits in perspective, I provide the following
stylized calibration exercise. I assume the firm in the empirical
application is small enough for the $20$ million euros to be the
penalty, a one-to-one exchange rate of euros to dollars, and a perceived
enforcement rate of $10\%$. With these simplifying assumptions, the
expected penalty of noncompliance is $2$ million dollars. In the
framework, I set $R_{BB}$ as $2$ million dollars and $R_{comp}$
to be zero. The firm now compares the expected profit loss ($\Pi_{BB}-\Pi_{comp}=$
$2.3$ million dollars) to the expected regulatory penalty ($R_{BB}-R_{comp}=$
$2$ million dollars).

From a regulatory perspective, the cost of compliance ($\$2.3$ million)
exceeds the expected penalty for non-compliance ($\$2$ million under
GDPR). This gap creates misaligned incentives for policy adoption.
Regulators could close this gap by increasing either penalty amounts
or enforcement rates. Notably, the subsequent EU's AI Act has increased
penalty amounts, potentially addressing this incentive gap.

On a flip side, ensuring compliance with the right-to-explanation
clause has a nontrivial impact on the firm's bottom line. The expected
profit loss of $\$2.3$ million is for one month of sales; scaled
up annually, that is $\$27.6$ million in lost profits if the firm
ran a promotional strategy every month. As a result, the impact of
right-to-explanation laws can be quite substantial and regulators
should consider these downstream impacts as they seek to implement
data and privacy laws in other jurisdictions. 

This calibration exercise can be readily extended to capture more
complex settings. The regulatory penalty can depend on the complexity
of the targeting policy as more complex targeting rules can face higher
enforcement rates.\footnote{\citet{Lambin2023} study an incomplete information game between firms
and regulators where firms can offer explainable algorithms or black
box algorithms and regulators can choose to audit the firm for compliance
with right-to-explanation laws.} Then, $R_{BB}$ can increase in the complexity of the black box.
For comprehensible policies that are have more than five clauses and
are not conversational, $R_{comp}$ can be set to increase with the
number of clauses in the comprehensible targeting policy. Other extensions
to this exercise can be added to tailor it to different scenarios.

Lastly, I showed the cost of explanation the calibration study for
one specific class of comprehensible policies that I proposed in Section
\ref{sec:Comprehensible-class}. This class of targeting sentences
conservatively complies with the right-to-explanation clause. Naturally,
other classes of comprehensible targeting policies can be considered
in the general framework and the calibration exercise can be repeated
with different classes of comprehensible targeting policies and black
box targeting policies.

\subsection{Long-term effects of comprehensibility \label{subsec:Long-term-effects-of-Comprehensibility}}

In many settings, offering a comprehensible policy can lead to downstream
benefits for the firm. Customers can learn from a comprehensible policy
but cannot learn from a black box. If the policy benefits customers,
they can learn how to get the treatment again which can build further
brand equity with the firm. On the flip side, customers who were excluded
from the treatment can understand why they were left out and what
they need to do in order to get the treatment.

For the firm's perspective, implementing a comprehensible policy is
simpler than implementing a black box policy. If the firm representatives
need to implement the targeting policy, then it is easier for them
to train and follow a comprehensible policy. For example, training
salespeople to follow a comprehensible policy will be simpler than
doing so for a black box policy. It is also easier for these salespeople
to explain the firm's policy to its customers. Comprehensible policies
are also easier to diagnose by the firm, and the firm can audit these
policies to ensure they do not use information from protected classes.

However, for some settings, comprehensibility of the targeting policy
may not be beneficial for customers and firms. Customers may even
dislike a transparent explanation of the algorithmic decision policy
in certain settings. For example, in online dating, an explanation
of the matchmaking algorithm may draw ire from customers. For firms
operating in a competitive environment, offering comprehensible targeting
policies can give competing firms insight about the firm's profitable
customer base and its decision making. In equilibrium, it may not
be beneficial to the firm to reveal such information to its competitors.

As a result, managers need to consider the long-term benefits and
costs of comprehensibility for the black box policy ($B_{BB}$) to
that of the comprehensible policy ($B_{comp}$). In many cases, it
seems that the benefits for comprehensibility are positive and are
long-term ($B_{comp}>B_{BB}$). Future research can explore the long-term
effect of comprehensibility for firms and customers.

\section{Conclusion\label{sec:Conclusion}}

Data and privacy regulations like GDPR and CCPA are swiftly gaining
traction worldwide. With GDPR, regulators in Europe now ask for a
``right to explanation'' where a black box algorithm's decisions
need to be explainable to customers by the firm's human representatives.
They have increasingly cracked down on firms violating data and privacy
laws, and proposals to expand the regulation have only increased.\footnote{Meta Platforms was fined 1.2 billion euros for not abiding by GDPR
rules on May 22, 2023 \citep{EDPB2023}. Although the GDPR violation
did not specifically pertain to the ``right to explanation\textquotedbl{}
clause, it highlights increasing enforcement of GDPR regulations.
Consequently, forward-looking firms should factor these regulations
into their decision making process.}

This paper provides a framework for firms to navigate right-to-explanation
laws. The framework constructs comprehensible targeting policies that
satisfy right-to-explanation requirements and quantifies the costs
of compliance. I propose a class of comprehensible policies that should
satisfy the new regulatory constraints and that takes on the form
of sentences. These sentences are conditional clauses linked by logic
operators. I further show how to find the optimal, profit-maximizing,
comprehensible policy for a sentence of given clause length. I then
benchmark the optimal comprehensible policy to the best-performing
black box targeting policy.

With the established framework, I document how the two targeting policies
differ and then quantify the cost of explanation, or the profit loss
from implementing the optimal comprehensible policy to the black box
policy. I provide an application for sending \$10-off promotions for
a durable goods retailer. I find the cost of explanation to be $23$
cents per person for a three-clause optimal comprehensible targeting
policy; this cost constitutes a $7.5\%$ loss in profits from the
black box policy.

I quantify the profit loss that the firm will face from complying
with right-to-explanation regulation. In the application, for a basis
of $10$ million customers, the cost of explanation leads to a $\$2.3$
million profit loss from the firm's promotional strategy. These losses
represent the economic impact on the firm when abiding by data and
privacy regulation. While GDPR fines have been mainly levied on large
multinational technology companies, my framework provides a localized
way for any company under its jurisdiction to quantify and evaluate
the regulation's impact on its bottom line. The recently enacted EU's
AI Act expands explainability requirements and increases penalties,
making such cost assessment increasingly critical for firms.

 This framework can be extended to capture benefits along with costs.
While I only provide a cost analysis of how comprehension in marketing
policies acts as a constraint on the firm's personalization and targeting
strategies, there may be benefits from comprehension for the firm's
customers. Customers appear to have a disdain for algorithmic decisions
in certain domains \citep{Dietvorst2015,Dietvorst2022,Yalcin2023},
and providing them a comprehensible explanation may lead them to foster
future goodwill toward the firm. I leave exploring the benefits of
comprehension to future research.

More generally, my framework enables firms to assess the impact of
practical marketing constraints on their objectives of interest. In
my setting, I use profits, or producer surplus, as the objective and
the constraint is the comprehensibility of the targeting policy. The
cost of explanation reflects how this constraint affects the firm\textquoteright s
profits. This framework has the flexibility to explore alternative
objectives like consumer surplus or total surplus and can accommodate
different constraints, such as privacy or fairness considerations. 

My paper links the theoretical targeting and personalization literature
to what is done in practice by accounting for regulatory constraints.
As black box algorithms gain more regulatory scrutiny with increasingly
widespread use of generative artificial intelligence (AI) models and
with the recently enacted AI Act, firms need to navigate the regulatory
environment if they decide to continue to leverage modern advances
in AI for their day-to-day operations. This paper assesses the cost
of right-to-explanation legislation for a firm's targeting policies.
Future research in evaluating the effects of rapidly expanding data
and privacy regulation on both firms and customers is encouraged
to further bridge theory and practice. 
\newpage{}

{\footnotesize\bibliographystyle{ecta}
\bibliography{./Bibliography/bib}

@article{Gillis2019,
author = {Gillis, Talia B. and Spiess, Jann L.},
journal = {University of Chicago Law Review},
number = {2},
pages = {459--487},
title = {{Big Data and Discrimination}},
url = {https://chicagounbound.uchicago.edu/uclrev/vol86/iss2/4/},
volume = {86},
year = {2019}
}

@article{Kitagawa2018,
author = {Kitagawa, Toru and Tetenov, Aleksey},
doi = {10.3982/ECTA13288},
issn = {0012-9682},
journal = {Econometrica},
number = {2},
pages = {591--616},
title = {{Who Should Be Treated? Empirical Welfare Maximization Methods for Treatment Choice}},
url = {https://www.econometricsociety.org/doi/10.3982/ECTA13288},
volume = {86},
year = {2018}
}

@book{Katsov2017,
author = {Katsov, Ilya},
isbn = {978-0692142608},
pages = {508},
title = {{Introduction to Algorithmic Marketing}},
year = {2017}
}

@article{Yoganarasimhan2020,
author = {Yoganarasimhan, Hema and Barzegary, Ebrahim and Pani, Abhishek},
doi = {10.1287/mnsc.2022.4507},
issn = {0025-1909},
journal = {Management Science},
month = {aug},
title = {{Design and Evaluation of Optimal Free Trials}},
url = {http://pubsonline.informs.org/doi/10.1287/mnsc.2022.4507},
year = {2022}
}

@article{Yalcin2023,
author = {Yalcin, Gizem and Themeli, Erlis and Stamhuis, Evert and Philipsen, Stefan and Puntoni, Stefano},
doi = {10.1007/s10506-022-09312-z},
issn = {0924-8463},
journal = {Artificial Intelligence and Law},
month = {jun},
number = {2},
pages = {269--292},
title = {{Perceptions of Justice By Algorithms}},
url = {https://link.springer.com/10.1007/s10506-022-09312-z},
volume = {31},
year = {2023}
}

@article{Rudin2019,
author = {Rudin, Cynthia},
doi = {10.1038/s42256-019-0048-x},
issn = {2522-5839},
journal = {Nature Machine Intelligence},
month = {may},
number = {5},
pages = {206--215},
title = {{Stop explaining black box machine learning models for high stakes decisions and use interpretable models instead}},
url = {https://www.nature.com/articles/s42256-019-0048-x},
volume = {1},
year = {2019}
}

@article{Wachter2016,
author = {Wachter, Sandra and Mittelstadt, Brent and Floridi, Luciano},
doi = {10.2139/ssrn.2903469},
issn = {1556-5068},
journal = {SSRN Electronic Journal},
title = {{Why a Right to Explanation of Automated Decision-Making Does Not Exist in the General Data Protection Regulation}},
url = {https://www.ssrn.com/abstract=2903469},
year = {2016}
}

@misc{EuropeanCommission2021,
author = {{European Commission}},
title = {{Regulation of the European parliament and of the Council laying down harmonised rules on artificial intelligence (Artificial Intelligence Act) and amending certain union legislative acts}},
year = {2021}
}

@article{Miller2019a,
author = {Miller, Tim},
doi = {10.1016/j.artint.2018.07.007},
issn = {00043702},
journal = {Artificial Intelligence},
month = {feb},
pages = {1--38},
title = {{Explanation in artificial intelligence: Insights from the social sciences}},
url = {https://linkinghub.elsevier.com/retrieve/pii/S0004370218305988},
volume = {267},
year = {2019}
}

@article{Angelino2018,
archivePrefix = {arXiv},
arxivId = {1704.01701},
author = {Angelino, Elaine and Larus-Stone, Nicholas and Alabi, Daniel and Seltzer, Margo and Rudin, Cynthia},
eprint = {1704.01701},
issn = {15337928},
journal = {Journal of Machine Learning Research},
pages = {1--78},
title = {{Learning certifiably optimal rule lists for categorical data}},
volume = {18},
year = {2018}
}

@article{Kleinberg2018,
author = {Kleinberg, Jon and Ludwig, Jens and Mullainathan, Sendhil and Sunstein, Cass R},
doi = {10.1093/jla/laz001},
issn = {1946-5319},
journal = {Journal of Legal Analysis},
month = {dec},
pages = {113--174},
title = {{Discrimination in the Age of Algorithms}},
url = {https://academic.oup.com/jla/article/doi/10.1093/jla/laz001/5476086},
volume = {10},
year = {2018}
}

@article{Mohammadi2024,
author = {Mohammadi, Behnam and Malik, Nikhil and Derdenger, Tim and Srinivasan, Kannan},
doi = {10.1287/mksc.2022.0396},
issn = {0732-2399},
journal = {Marketing Science},
month = {nov},
title = {{Regulating Explainable Artificial Intelligence (XAI) May Harm Consumers}},
url = {https://pubsonline.informs.org/doi/10.1287/mksc.2022.0396},
year = {2024}
}

@book{BSEthics,
address = {Princeton},
author = {Spinoza, B},
editor = {Curley, E},
publisher = {Princeton University Press},
title = {{The Collected Works of Spinoza}},
volume = {1},
year = {1985}
}

@article{Rai2020,
author = {Rai, Arun},
doi = {10.1007/s11747-019-00710-5},
issn = {0092-0703},
journal = {Journal of the Academy of Marketing Science},
month = {jan},
number = {1},
pages = {137--141},
title = {{Explainable AI: from black box to glass box}},
url = {http://link.springer.com/10.1007/s11747-019-00710-5},
volume = {48},
year = {2020}
}

@book{Imbens2015,
author = {Imbens, Guido W. and Rubin, Donald B.},
doi = {10.1017/CBO9781139025751},
isbn = {9780521885881},
month = {apr},
publisher = {Cambridge University Press},
title = {{Causal Inference for Statistics, Social, and Biomedical Sciences}},
url = {https://www.cambridge.org/core/product/identifier/9781139025751/type/book},
year = {2015}
}

@article{Edwards2017,
author = {Edwards, Lilian and Veale, Michael},
doi = {10.2139/ssrn.2972855},
issn = {1556-5068},
journal = {SSRN Electronic Journal},
title = {{Slave to the Algorithm? Why a Right to Explanation is Probably Not the Remedy You are Looking for}},
url = {https://www.ssrn.com/abstract=2972855},
year = {2017}
}

@article{Smith2022,
author = {Smith, Adam N. and Seiler, Stephan and Aggarwal, Ishant},
doi = {10.1287/mksc.2022.1387},
issn = {0732-2399},
journal = {Marketing Science},
month = {aug},
title = {{Optimal Price Targeting}},
url = {http://pubsonline.informs.org/doi/10.1287/mksc.2022.1387},
year = {2022}
}

@article{Ellickson2022,
author = {Ellickson, Paul B. and Kar, Wreetabrata and Reeder, James C.},
doi = {10.1287/mksc.2022.1401},
issn = {0732-2399},
journal = {Marketing Science},
month = {sep},
title = {{Estimating Marketing Component Effects: Double Machine Learning from Targeted Digital Promotions}},
url = {http://pubsonline.informs.org/doi/10.1287/mksc.2022.1401},
year = {2022}
}

@book{DellaRocca2008,
address = {New York, NY},
author = {{Della Rocca}, Michael},
edition = {1},
isbn = {0-203-89458-8},
title = {{Spinoza}},
year = {2008}
}

@article{Zhang2022,
archivePrefix = {arXiv},
arxivId = {2204.05793},
author = {Zhang, Walter W. and Misra, Sanjog},
eprint = {2204.05793},
month = {apr},
title = {{Coarse Personalization}},
url = {http://arxiv.org/abs/2204.05793},
year = {2022}
}

@article{Rossi1996,
author = {Rossi, Peter E. and McCulloch, Robert E. and Allenby, Greg M.},
doi = {10.1287/mksc.15.4.321},
issn = {07322399},
journal = {Marketing Science},
number = {4},
pages = {321--340},
title = {{The value of purchase history data in target marketing}},
volume = {15},
year = {1996}
}

@article{Halpern2005b,
author = {Halpern, Joseph Y. and Pearl, Judea},
doi = {10.1093/bjps/axi147},
issn = {0007-0882},
journal = {The British Journal for the Philosophy of Science},
month = {dec},
number = {4},
pages = {843--887},
title = {{Causes and Explanations: A Structural-Model Approach. Part I: Causes}},
url = {https://www.journals.uchicago.edu/doi/10.1093/bjps/axi147},
volume = {56},
year = {2005}
}

@article{Lipton1990,
author = {Lipton, Peter},
doi = {10.1017/S1358246100005130},
issn = {1358-2461},
journal = {Royal Institute of Philosophy Supplement},
month = {mar},
pages = {247--266},
title = {{Contrastive Explanation}},
url = {https://www.cambridge.org/core/product/identifier/S1358246100005130/type/journal_article},
volume = {27},
year = {1990}
}

@article{Senoner2022,
author = {Senoner, Julian and Netland, Torbj{\o}rn and Feuerriegel, Stefan},
doi = {10.1287/mnsc.2021.4190},
issn = {0025-1909},
journal = {Management Science},
month = {aug},
number = {8},
pages = {5704--5723},
title = {{Using Explainable Artificial Intelligence to Improve Process Quality: Evidence from Semiconductor Manufacturing}},
url = {http://pubsonline.informs.org/doi/10.1287/mnsc.2021.4190},
volume = {68},
year = {2022}
}

@article{Farrell2021,
archivePrefix = {arXiv},
arxivId = {1809.09953},
author = {Farrell, Max H. and Liang, Tengyuan and Misra, Sanjog},
doi = {10.3982/ecta16901},
eprint = {1809.09953},
issn = {0012-9682},
journal = {Econometrica},
month = {sep},
number = {1},
pages = {181--213},
title = {{Deep Neural Networks for Estimation and Inference}},
url = {http://arxiv.org/abs/1809.09953 https://www.econometricsociety.org/doi/10.3982/ECTA16901},
volume = {89},
year = {2021}
}

@misc{EuropeanCommission2016,
author = {{European Commission}},
pages = {1--88},
publisher = {European Commission},
title = {{Regulation (EU) 2016/679 of the European Parliament and of the Council of 27 April 2016 on the protection of natural persons with regard to the processing of personal data and on the free movement of such data, and repealing Directive 95/46/EC}},
url = {https://ec.europa.eu/commission/sites/beta-political/files/data-protection-factsheet-changes_en.pdf},
year = {2016}
}

@article{Fong2021,
author = {Fong, Hortense and Kumar, Vineet and Sudhir, K.},
doi = {10.2139/ssrn.4025386},
issn = {1556-5068},
journal = {SSRN Electronic Journal},
title = {{A Theory-Based Interpretable Deep Learning Architecture for Music Emotion}},
url = {https://www.ssrn.com/abstract=4025386},
year = {2021}
}

@article{Wager2018,
archivePrefix = {arXiv},
arxivId = {1510.04342},
author = {Wager, Stefan and Athey, Susan},
doi = {10.1080/01621459.2017.1319839},
eprint = {1510.04342},
issn = {1537274X},
journal = {Journal of the American Statistical Association},
month = {jul},
number = {523},
pages = {1228--1242},
title = {{Estimation and Inference of Heterogeneous Treatment Effects using Random Forests}},
url = {https://www.tandfonline.com/doi/full/10.1080/01621459.2017.1319839},
volume = {113},
year = {2018}
}

@article{Rafieian2022,
author = {Rafieian, Omid and Yoganarasimhan, Hema},
doi = {10.2139/ssrn.4123356},
issn = {1556-5068},
journal = {SSRN Electronic Journal},
title = {{AI and Personalization}},
url = {https://www.ssrn.com/abstract=4123356},
year = {2022}
}

@article{Zhang2012,
author = {Zhang, Baqun and Tsiatis, Anastasios A. and Davidian, Marie and Zhang, Min and Laber, Eric},
doi = {10.1002/sta.411},
issn = {20491573},
journal = {Stat},
month = {oct},
number = {1},
pages = {103--114},
title = {{Estimating optimal treatment regimes from a classification perspective}},
url = {https://onlinelibrary.wiley.com/doi/10.1002/sta.411},
volume = {1},
year = {2012}
}

@article{Simester2020,
author = {Simester, Duncan and Timoshenko, Artem and Zoumpoulis, Spyros I.},
doi = {10.1287/mnsc.2019.3379},
issn = {15265501},
journal = {Management Science},
month = {aug},
number = {8},
pages = {3412--3424},
title = {{Efficiently evaluating targeting policies: Improving on champion vs. Challenger experiments}},
url = {http://pubsonline.informs.org/doi/10.1287/mnsc.2019.3379},
volume = {66},
year = {2020}
}

@article{Hauser2010,
author = {Hauser, John R. and Toubia, Olivier and Evgeniou, Theodoros and Befurt, Rene and Dzyabura, Daria},
doi = {10.1509/jmkr.47.3.485},
issn = {0022-2437},
journal = {Journal of Marketing Research},
month = {jun},
number = {3},
pages = {485--496},
title = {{Disjunctions of Conjunctions, Cognitive Simplicity, and Consideration Sets}},
url = {https://journals.sagepub.com/doi/10.1509/jmkr.47.3.485},
volume = {47},
year = {2010}
}

@book{Lilien1992,
author = {Lilien, Gary L. and Kotler, Philip and Moorthy, K. Sridhar},
isbn = {0135446449},
publisher = {Prentice-Hall},
title = {{Marketing Models}},
year = {1992}
}

@inproceedings{Cawsey1991a,
author = {Cawsey, Alison},
booktitle = {Ninth National Conference on Artificial Intelligence},
pages = {86--91},
title = {{Generating Interactive Explanations}},
url = {http://www.aaai.org/Papers/AAAI/1991/AAAI91-014.pdf},
year = {1991}
}

@article{Lambin2023,
author = {Lambin, Xavier and Raizonville, Adrien},
doi = {10.2139/ssrn.4455428},
journal = {SSRN Electronic Journal},
title = {{From Black Box to Glass Box: Algorithmic Explainability as a Strategic Decision}},
year = {2023}
}

@book{Breiman1984,
author = {Breiman, Leo},
doi = {10.1201/9781315139470},
edition = {1},
isbn = {9781315139470},
month = {oct},
publisher = {Routledge},
title = {{Classification And Regression Trees}},
url = {https://www.taylorfrancis.com/books/9781351460491},
year = {1984}
}

@book{Goodfellow-et-al-2016,
author = {Goodfellow, Ian and Bengio, Yoshua and Courville, Aaron},
publisher = {MIT Press},
title = {{Deep Learning}},
year = {2016}
}

@article{VanOsselaer2000,
author = {van Osselaer, Stijn M. J. and Alba, Joseph W.},
doi = {10.1086/314305},
issn = {0093-5301},
journal = {Journal of Consumer Research},
month = {jun},
number = {1},
pages = {1--16},
title = {{Consumer Learning and Brand Equity}},
url = {https://academic.oup.com/jcr/article-lookup/doi/10.1086/314305},
volume = {27},
year = {2000}
}

@misc{EDPB2023,
author = {{European Data Protection Board}},
month = {may},
title = {{1.2 billion euro fine for Facebook as a result of EDPB binding decision}},
url = {https://edpb.europa.eu/news/news/2023/12-billion-euro-fine-facebook-result-edpb-binding-decision_en},
year = {2023}
}

@article{Hitsch2018,
author = {Hitsch, Guenter J. and Misra, Sanjog and Zhang, Walter W.},
doi = {10.2139/ssrn.3111957},
issn = {1556-5068},
journal = {SSRN Electronic Journal},
title = {{Heterogeneous Treatment Effects and Optimal Targeting Policy Evaluation}},
url = {https://www.ssrn.com/abstract=3111957},
year = {2023}
}

@article{Wang2023,
author = {Wang, Qiaochu and Huang, Yan and Jasin, Stefanus and Singh, Param Vir},
doi = {10.1287/mnsc.2022.4475},
issn = {0025-1909},
journal = {Management Science},
month = {apr},
number = {4},
pages = {2297--2317},
title = {{Algorithmic Transparency with Strategic Users}},
url = {https://pubsonline.informs.org/doi/10.1287/mnsc.2022.4475},
volume = {69},
year = {2023}
}

@article{Cawsey1993a,
author = {Cawsey, Alison},
doi = {10.1006/imms.1993.1009},
issn = {00207373},
journal = {International Journal of Man-Machine Studies},
month = {feb},
number = {2},
pages = {169--199},
title = {{Planning interactive explanations}},
url = {https://linkinghub.elsevier.com/retrieve/pii/S0020737383710096},
volume = {38},
year = {1993}
}

@article{Weiner1980a,
author = {Weiner, J.L.},
doi = {10.1016/0004-3702(80)90021-1},
issn = {00043702},
journal = {Artificial Intelligence},
month = {nov},
number = {1-2},
pages = {19--48},
title = {{BLAH, a system which explains its reasoning}},
url = {https://linkinghub.elsevier.com/retrieve/pii/0004370280900211},
volume = {15},
year = {1980}
}

@article{Athey2016,
archivePrefix = {arXiv},
arxivId = {1504.01132},
author = {Athey, Susan and Imbens, Guido},
doi = {10.1073/pnas.1510489113},
eprint = {1504.01132},
issn = {10916490},
journal = {Proceedings of the National Academy of Sciences of the United States of America},
month = {jul},
number = {27},
pages = {7353--7360},
pmid = {27382149},
title = {{Recursive partitioning for heterogeneous causal effects}},
url = {http://www.pnas.org/lookup/doi/10.1073/pnas.1510489113},
volume = {113},
year = {2016}
}

@article{Ni2012,
author = {Ni, Jian and Neslin, Scott A. and Sun, Baohong},
doi = {10.1287/mksc.1120.0726},
issn = {0732-2399},
journal = {Marketing Science},
month = {nov},
number = {6},
pages = {1008--1013},
title = {{Database Submission The ISMS Durable Goods Data Sets}},
url = {https://pubsonline.informs.org/doi/10.1287/mksc.1120.0726},
volume = {31},
year = {2012}
}

@article{Hitsch2024,
author = {Hitsch, G{\"{u}}nter J. and Misra, Sanjog and Zhang, Walter W.},
doi = {10.1007/s11129-023-09278-5},
issn = {1570-7156},
journal = {Quantitative Marketing and Economics},
month = {apr},
title = {{Heterogeneous treatment effects and optimal targeting policy evaluation}},
url = {https://link.springer.com/10.1007/s11129-023-09278-5},
year = {2024}
}

@article{Kleinberg2017,
author = {Kleinberg, Jon and Lakkaraju, Himabindu and Leskovec, Jure and Ludwig, Jens and Mullainathan, Sendhil},
doi = {10.1093/qje/qjx032},
issn = {0033-5533},
journal = {The Quarterly Journal of Economics},
month = {aug},
title = {{Human Decisions and Machine Predictions*}},
url = {http://academic.oup.com/qje/article/doi/10.1093/qje/qjx032/4095198/Human-Decisions-and-Machine-Predictions},
year = {2017}
}

@book{Taddy2019,
author = {Taddy, Matt},
edition = {1},
isbn = {9781260452778},
pages = {303--307},
publisher = {McGraw Hill},
title = {{Business data science : combining machine learning and economics to optimize, automate, and accelerate business decisions}},
year = {2019}
}

@article{Fernandez-Loria2023,
author = {Fern{\'{a}}ndez-Lor{\'{i}}a, Carlos and Provost, Foster and Anderton, Jesse and Carterette, Benjamin and Chandar, Praveen},
doi = {10.1287/isre.2022.1149},
issn = {1047-7047},
journal = {Information Systems Research},
month = {jun},
number = {2},
pages = {786--803},
title = {{A Comparison of Methods for Treatment Assignment with an Application to Playlist Generation}},
url = {https://pubsonline.informs.org/doi/10.1287/isre.2022.1149},
volume = {34},
year = {2023}
}

@article{Kleinberg2015,
author = {Kleinberg, Jon and Ludwig, Jens and Mullainathan, Sendhil and Obermeyer, Ziad},
doi = {10.1257/aer.p20151023},
issn = {0002-8282},
journal = {American Economic Review},
month = {may},
number = {5},
pages = {491--495},
title = {{Prediction Policy Problems}},
url = {https://pubs.aeaweb.org/doi/10.1257/aer.p20151023},
volume = {105},
year = {2015}
}

@book{Vapnik2000,
address = {New York, NY},
author = {Vapnik, Vladimir N.},
doi = {10.1007/978-1-4757-3264-1},
isbn = {978-1-4419-3160-3},
publisher = {Springer New York},
title = {{The Nature of Statistical Learning Theory}},
url = {http://link.springer.com/10.1007/978-1-4757-3264-1},
year = {2000}
}

@article{Wang2017,
author = {Wang, Tong and Rudin, Cynthia and Doshi-Velez, Finale and Liu, Yimin and Klampfl, Erica and MacNeille, Perry},
issn = {15337928},
journal = {Journal of Machine Learning Research},
pages = {1--37},
title = {{A Bayesian framework for learning rule sets for interpretable classification}},
volume = {18},
year = {2017}
}

@article{Dietvorst2015,
author = {Dietvorst, Berkeley J. and Simmons, Joseph P. and Massey, Cade},
doi = {10.1037/xge0000033},
issn = {1939-2222},
journal = {Journal of Experimental Psychology: General},
number = {1},
pages = {114--126},
title = {{Algorithm aversion: People erroneously avoid algorithms after seeing them err.}},
url = {http://doi.apa.org/getdoi.cfm?doi=10.1037/xge0000033},
volume = {144},
year = {2015}
}

@book{Hastie2015,
author = {Hastie, Trevor and Tibshirani, Robert and Wainwright, Martin},
doi = {10.1201/b18401},
edition = {1},
isbn = {9781498712170},
pages = {1--337},
publisher = {CRC Press},
title = {{Statistical learning with sparsity: The lasso and generalizations}},
year = {2015}
}

@incollection{Schwender2010,
author = {Schwender, Holger and Ruczinski, Ingo},
doi = {10.1016/B978-0-12-380862-2.00002-3},
pages = {25--45},
title = {{Logic Regression and Its Extensions}},
url = {https://linkinghub.elsevier.com/retrieve/pii/B9780123808622000023},
year = {2010}
}

@article{Athey2021,
archivePrefix = {arXiv},
arxivId = {1702.02896},
author = {Athey, Susan and Wager, Stefan},
doi = {10.3982/ecta15732},
eprint = {1702.02896},
issn = {0012-9682},
journal = {Econometrica},
number = {1},
pages = {133--161},
title = {{Policy Learning With Observational Data}},
url = {https://www.econometricsociety.org/doi/10.3982/ECTA15732},
volume = {89},
year = {2021}
}

@incollection{Zhang2024,
author = {Zhang, Walter W},
booktitle = {Dissertation},
chapter = {3},
organization = {University of Chicago},
pages = {16--21},
title = {{Policy learning with deep neural networks}},
year = {2024}
}

@misc{Falbel2023,
author = {Falbel, Daniel and Luraschi, Javier},
title = {{torch: Tensors and Neural Networks with 'GPU' Acceleration}},
year = {2023}
}

@article{Halpern2005c,
author = {Halpern, Joseph Y. and Pearl, Judea},
doi = {10.1093/bjps/axi148},
issn = {0007-0882},
journal = {The British Journal for the Philosophy of Science},
month = {dec},
number = {4},
pages = {889--911},
title = {{Causes and Explanations: A Structural-Model Approach. Part II: Explanations}},
url = {https://www.journals.uchicago.edu/doi/10.1093/bjps/axi148},
volume = {56},
year = {2005}
}

@article{Wang2025,
archivePrefix = {arXiv},
arxivId = {2502.16759},
author = {Wang, Yuyan and Li, Pan and Chen, Minmin},
eprint = {2502.16759},
month = {feb},
title = {{The Blessing of Reasoning: LLM-Based Contrastive Explanations in Black-Box Recommender Systems}},
url = {http://arxiv.org/abs/2502.16759},
year = {2025}
}

@article{Biran2017a,
author = {Biran, Or and Cotton, Courtenay},
journal = {IJCAI 2017 Workshop on Explainable Artificial Intelligence (XAI)},
pages = {8--13},
title = {{Explanation and Justification in Machine Learning: A Survey}},
year = {2017}
}

@article{Dietvorst2022,
author = {Dietvorst, Berkeley J. and Bartels, Daniel M.},
doi = {10.1002/jcpy.1266},
issn = {1057-7408},
journal = {Journal of Consumer Psychology},
month = {jul},
number = {3},
pages = {406--424},
title = {{Consumers Object to Algorithms Making Morally Relevant Tradeoffs Because of Algorithms' Consequentialist Decision Strategies}},
url = {https://onlinelibrary.wiley.com/doi/10.1002/jcpy.1266},
volume = {32},
year = {2022}
}

@article{Chintagunta2023,
author = {Chintagunta, Pradeep K. and Huang, Liqiang and Miao, Wei and Zhang, Wanqing},
doi = {10.2139/ssrn.4423917},
issn = {1556-5068},
journal = {SSRN Electronic Journal},
title = {{Measuring Seller Response to Buyer-initiated Disintermediation: Evidence from a Field Experiment on a Service Platform}},
url = {https://www.ssrn.com/abstract=4423917},
year = {2023}
}

@book{Cawsey1992a,
author = {Cawsey, Alison},
isbn = {0262032023},
publisher = {MIT Press},
title = {{Explanation and Interaction: The Computer Generation of Explanatory Dialogues}},
year = {1992}
}

@article{Karlinsky-Shichor2019,
author = {Karlinsky-Shichor, Yael and Netzer, Oded},
doi = {10.2139/ssrn.3368402},
issn = {1556-5068},
journal = {SSRN Electronic Journal},
title = {{Automating the B2B Salesperson Pricing Decisions: Can Machines Replace Humans and When?}},
url = {https://www.ssrn.com/abstract=3368402},
year = {2019}
}

@article{Ascarza2018,
author = {Ascarza, Eva},
doi = {10.1509/jmr.16.0163},
issn = {0022-2437},
journal = {Journal of Marketing Research},
month = {feb},
number = {1},
pages = {80--98},
title = {{Retention Futility: Targeting High-Risk Customers Might be Ineffective}},
url = {http://journals.sagepub.com/doi/10.1509/jmr.16.0163},
volume = {55},
year = {2018}
}

@article{Wang2022,
author = {Wang, Tong and He, Cheng and Jin, Fujie and Hu, Yu Jeffrey},
doi = {10.1287/isre.2021.1078},
issn = {1047-7047},
journal = {Information Systems Research},
month = {jun},
number = {2},
pages = {659--677},
title = {{Evaluating the Effectiveness of Marketing Campaigns for Malls Using a Novel Interpretable Machine Learning Model}},
url = {https://pubsonline.informs.org/doi/10.1287/isre.2021.1078},
volume = {33},
year = {2022}
}

@article{Nemhauser1978,
author = {Nemhauser, G. L. and Wolsey, L. A. and Fisher, M. L.},
doi = {10.1007/BF01588971},
issn = {0025-5610},
journal = {Mathematical Programming},
month = {dec},
number = {1},
pages = {265--294},
title = {{An analysis of approximations for maximizing submodular set functions-I}},
url = {http://link.springer.com/10.1007/BF01588971},
volume = {14},
year = {1978}
}

@article{Mothilal2019,
archivePrefix = {arXiv},
arxivId = {1905.07697},
author = {Mothilal, Ramaravind Kommiya and Sharma, Amit and Tan, Chenhao},
doi = {10.1145/3351095.3372850},
eprint = {1905.07697},
month = {may},
title = {{Explaining Machine Learning Classifiers through Diverse Counterfactual Explanations}},
url = {http://arxiv.org/abs/1905.07697 http://dx.doi.org/10.1145/3351095.3372850},
year = {2019}
}

@article{Zhao2012,
author = {Zhao, Yingqi and Zeng, Donglin and Rush, A. John and Kosorok, Michael R.},
doi = {10.1080/01621459.2012.695674},
issn = {0162-1459},
journal = {Journal of the American Statistical Association},
month = {sep},
number = {499},
pages = {1106--1118},
title = {{Estimating Individualized Treatment Rules Using Outcome Weighted Learning}},
url = {https://www.tandfonline.com/doi/full/10.1080/01621459.2012.695674},
volume = {107},
year = {2012}
}

@inproceedings{NIPS2017_8a20a862,
author = {Lundberg, Scott M and Lee, Su-In},
booktitle = {Advances in Neural Information Processing Systems},
editor = {Guyon, I and Luxburg, U Von and Bengio, S and Wallach, H and Fergus, R and Vishwanathan, S and Garnett, R},
publisher = {Curran Associates, Inc.},
title = {{A Unified Approach to Interpreting Model Predictions}},
url = {https://proceedings.neurips.cc/paper/2017/file/8a20a8621978632d76c43dfd28b67767-Paper.pdf},
volume = {30},
year = {2017}
}

@article{Bian2017,
author = {Bian, A.A. and Buhmann, J.M. and Krause, A. and Tschiatschek, S.},
journal = {34th International Conference on Machine Learning, ICML 2017},
pages = {756},
title = {{Guarantees for greedy maximization of non-submodular functions with applications}},
volume = {1},
year = {2017}
}

@inproceedings{Ribeiro2016,
address = {New York, NY, USA},
author = {Ribeiro, Marco Tulio and Singh, Sameer and Guestrin, Carlos},
booktitle = {Proceedings of the 22nd ACM SIGKDD International Conference on Knowledge Discovery and Data Mining},
doi = {10.1145/2939672.2939778},
isbn = {9781450342322},
month = {aug},
pages = {1135--1144},
publisher = {ACM},
title = {{"Why Should I Trust You?"}},
url = {https://dl.acm.org/doi/10.1145/2939672.2939778},
year = {2016}
}

@article{Athey2019,
archivePrefix = {arXiv},
arxivId = {1610.01271},
author = {Athey, Susan and Tibshirani, Julie and Wager, Stefan},
doi = {10.1214/18-AOS1709},
eprint = {1610.01271},
issn = {00905364},
journal = {Annals of Statistics},
month = {apr},
number = {2},
pages = {1179--1203},
title = {{Generalized random forests}},
url = {https://projecteuclid.org/euclid.aos/1547197251},
volume = {47},
year = {2019}
}
}{\footnotesize\par}

\newpage{}

\section*{Figures}

\begin{figure}[H]
\caption{Methodological overview\label{fig:Methodological-overview}}

\begin{centering}
\includegraphics[width=0.9\textwidth]{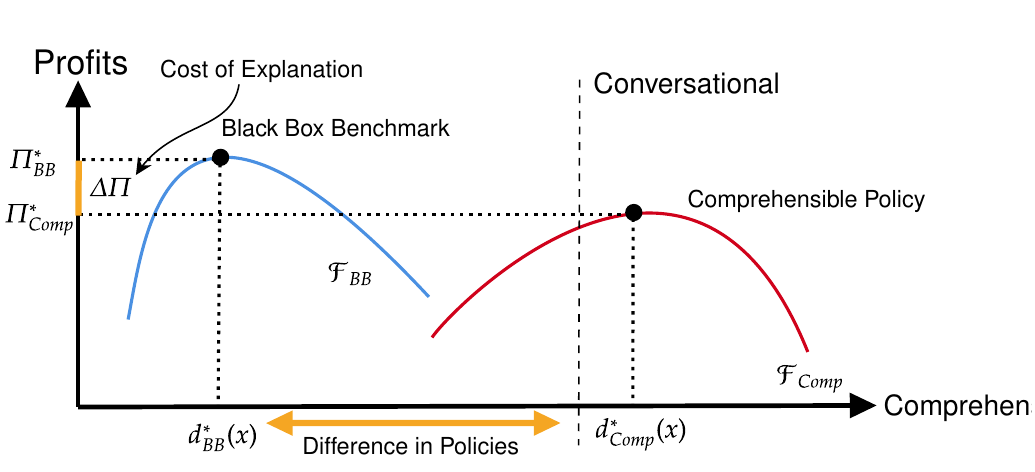}
\par\end{centering}
{\small Note: This figure provides an overview of this paper's methodology
and it is described in Section \ref{sec:Framework}. The axes demonstrate
the trade-off between short-term profitability and comprehensibility.
Comprehensibility can be thought of $1/(\text{Model Complexity})$.
The blue curve on the left represents a class of black box targeting
policies ($\mathcal{F}_{BB}$) with a optimal targeting policy $d_{BB}^{*}(x)$.
Section \ref{sec:Comprehensible-class} traces out a class of comprehensible
policies ($\mathcal{F}_{Comp}$) with the red curve on the right.
Section \ref{sec:Optimal-comprehensible-policy} shows how to find
the optimal comprehensible targeting policy $d_{\text{Comp}}^{*}(x)$
among the set of comprehensible policies. In the empirical application,
Section \ref{subsec:Targeting-Policy-Difference} documents the differences
in $d_{DNN}^{*}(x)$ and $d_{\text{Comp}}^{*}(x)$ which is denoted
here as the differences in policies. Section \ref{subsec:Targeting-Policy-Profit-Difference}
describes the profits differences of $\Delta\Pi=\Pi_{DNN}^{*}-\Pi_{\text{Comp}}^{*}$
which is denoted here as the cost of explanation.}{\small\par}
\end{figure}

\begin{figure}[H]
\begin{centering}
\caption{Distribution of past November sales\label{fig:Generating-Clauses-from-Continuous-Variables}}
\includegraphics[width=0.65\paperwidth]{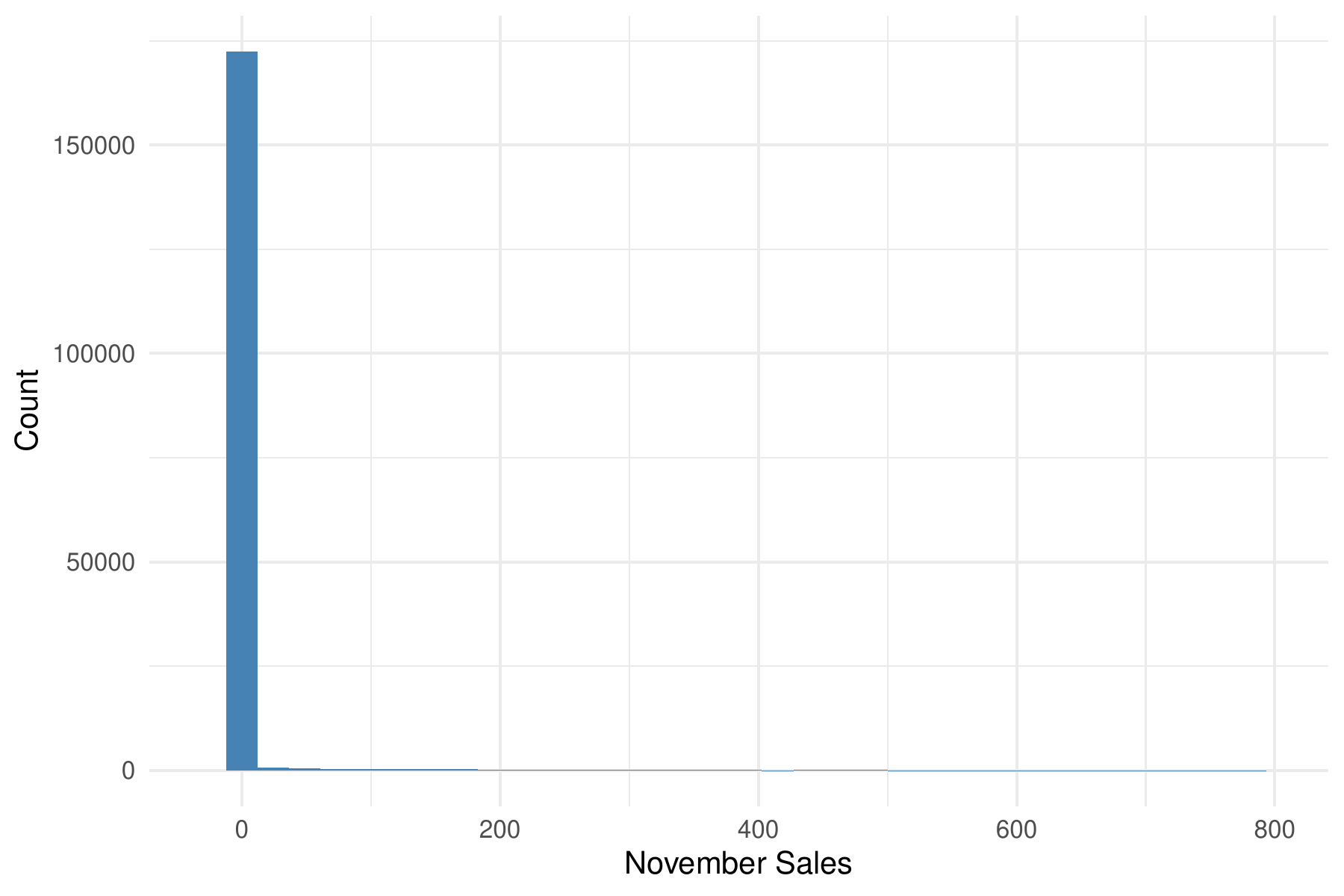}
\par\end{centering}
\begin{centering}
\includegraphics[width=0.65\paperwidth]{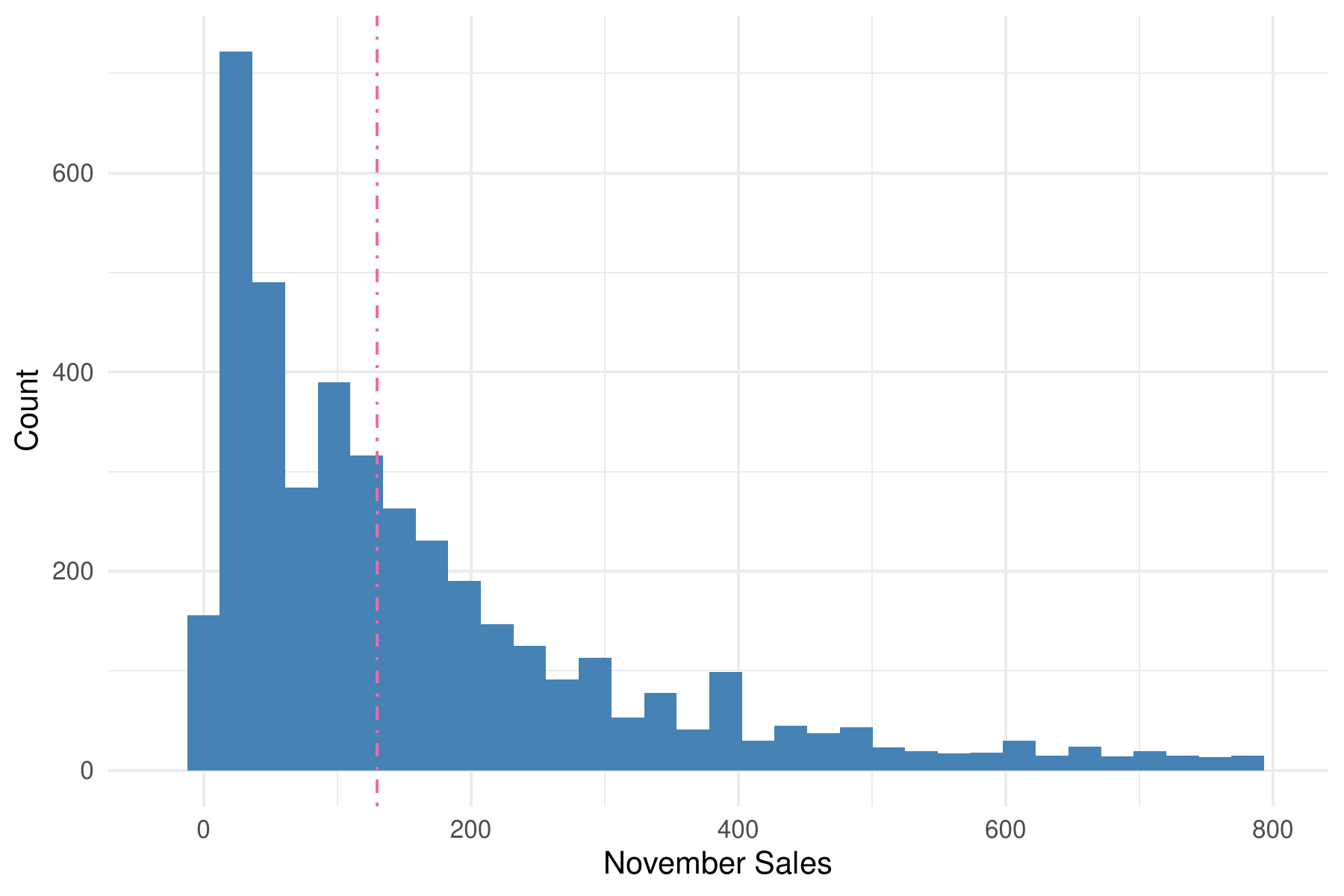}
\par\end{centering}
{\small Note: These two figures plot the distribution of past November
sales variable from the empirical application's RFM dataset. The upper
panel shows the unconditional distribution in which 97\% of the customers
do not buy anything during November. The lower panel shows the conditional
distribution of customers with non-zero spend during November. The
median spend for the conditional distribution is $\$129.99$ and is
denoted by the vertical dashed line.}{\small\par}
\end{figure}

\begin{figure}[H]
\begin{centering}
\caption{Comprehensible policy targeting percentages by clauses\label{fig:Targeting-Percentage}}
\includegraphics[width=0.45\paperwidth]{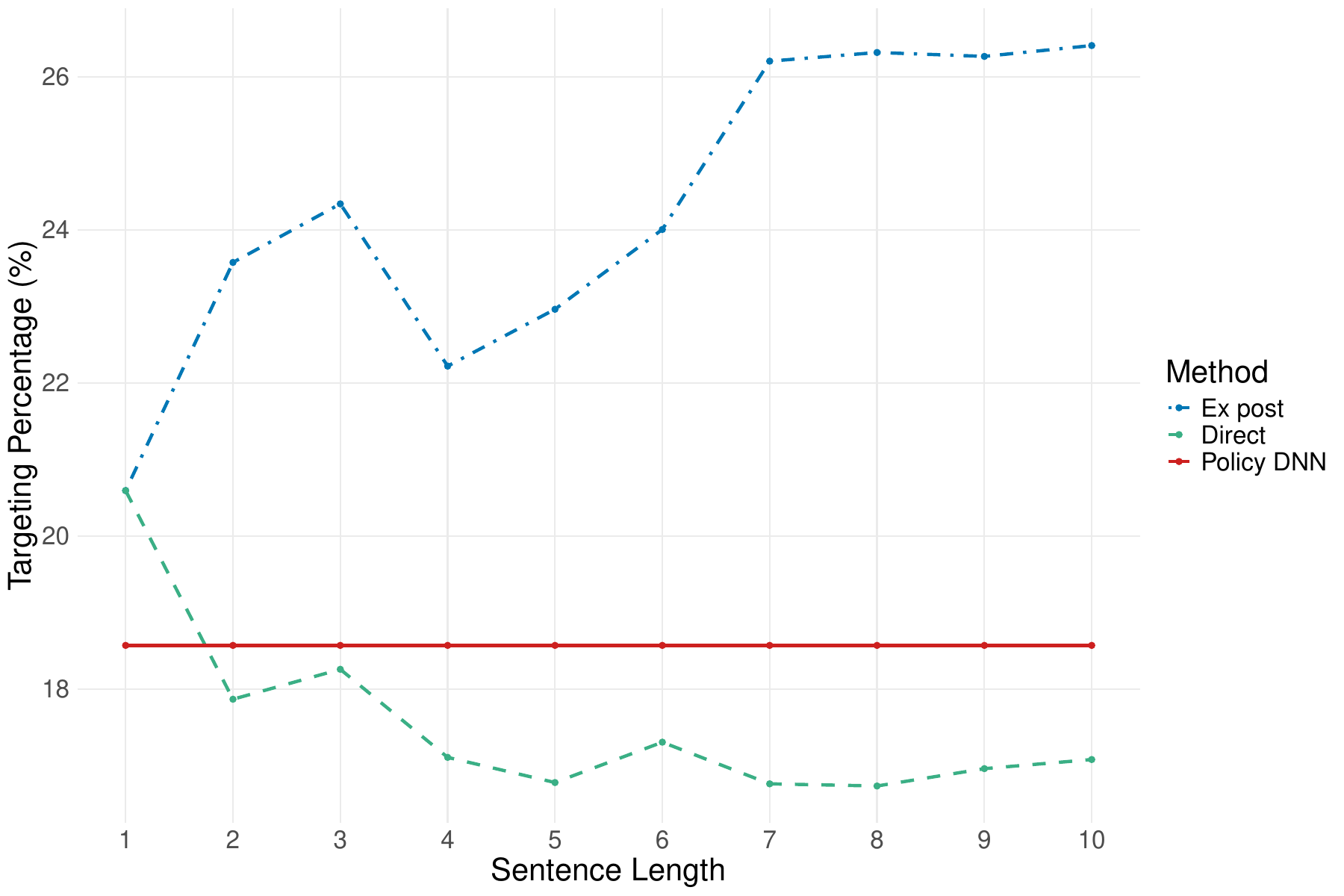}
\par\end{centering}
{\small Note: This figure plots the targeting percentage of the optimal
comprehensible policies by number of clauses in sample. The }{\small\emph{direct}}{\small{}
method learns the comprehensible policy from the data and the }{\small\emph{ex
post}}{\small{} method projects down the Policy DNN down to a comprehensible
policy. The Policy DNN is represented by the dashed horizontal line
and targets $18.1\%$ of customers.}{\small\par}
\end{figure}

\begin{figure}[H]
\begin{centering}
\caption{Two clause comprehensible policy vs. Policy DNN policy \label{fig:Two-Clause-Diagram}}
\includegraphics[width=0.45\paperwidth]{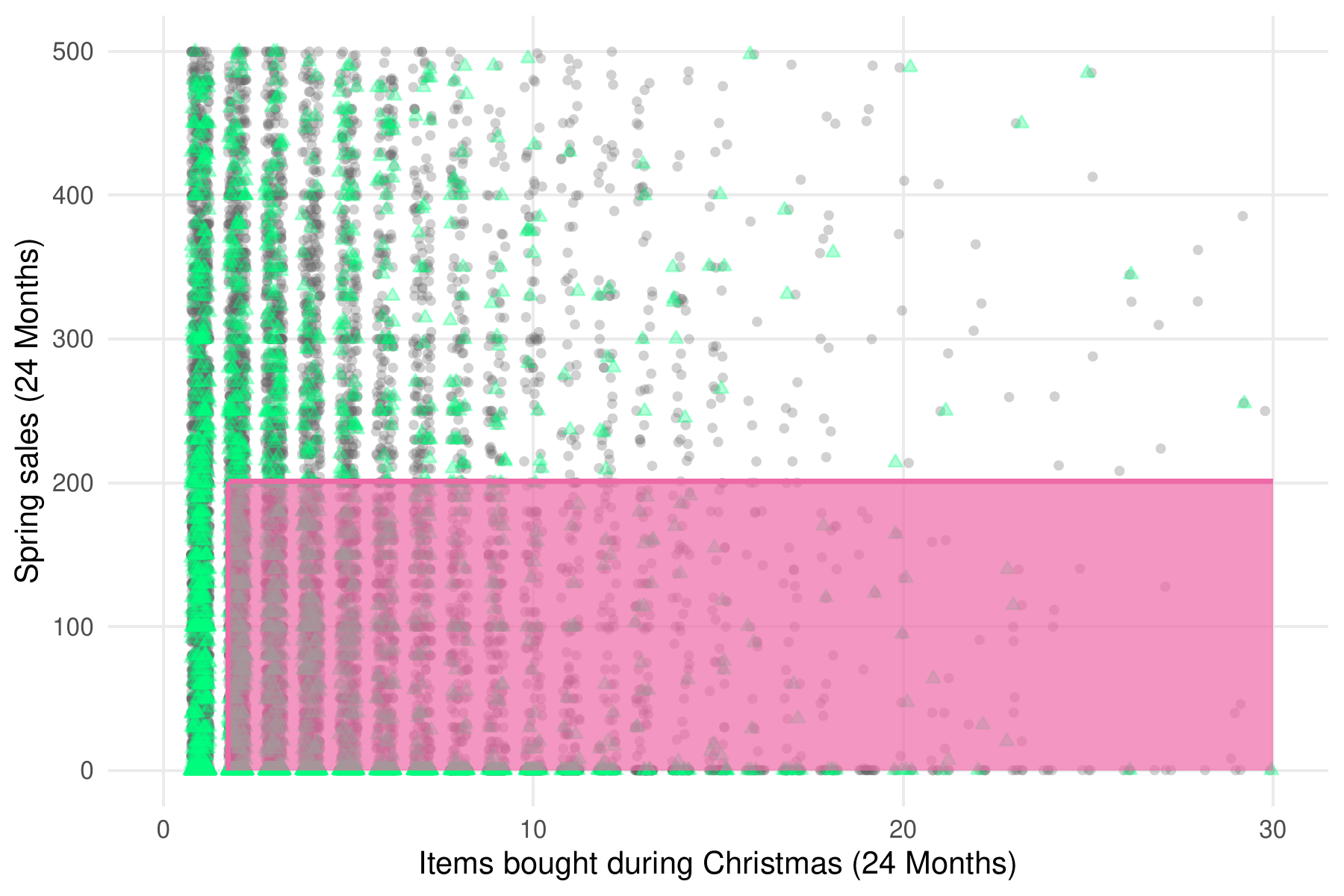}
\par\end{centering}
{\small Note: This figure contrasts the two-clause comprehensible targeting
policy to the Policy DNN targeting policy. The jittered grey circular
points represents the data observations in the raw data for the RFM
covariates spring sales (24 months) and items bought during Christmas
(24 months). The jittered green triangular points represents which
of customers the Policy DNN targets. The points within the pink rectangle
will be targeted by the two-clause comprehensible policy, which is,
``Target if customer has high amount of items during Christmas over
the last two years and did not have high Spring spending over the
last two years''. High amount of Christmas items over two years is
buying two or more items. High spending in spring over the last two
years is spending at least $\$200$. The optimal comprehensible policy
is formed using the }{\small\emph{direct}}{\small{} method that learns
the comprehensible policy from the data.}{\small\par}
\end{figure}

\begin{figure}[H]
\begin{centering}
\caption{Individual expected profits by method \label{fig:Individual-Expected-Profits}}
\subfloat[In sample expected profits]{\centering{}\includegraphics[width=0.6\paperwidth]{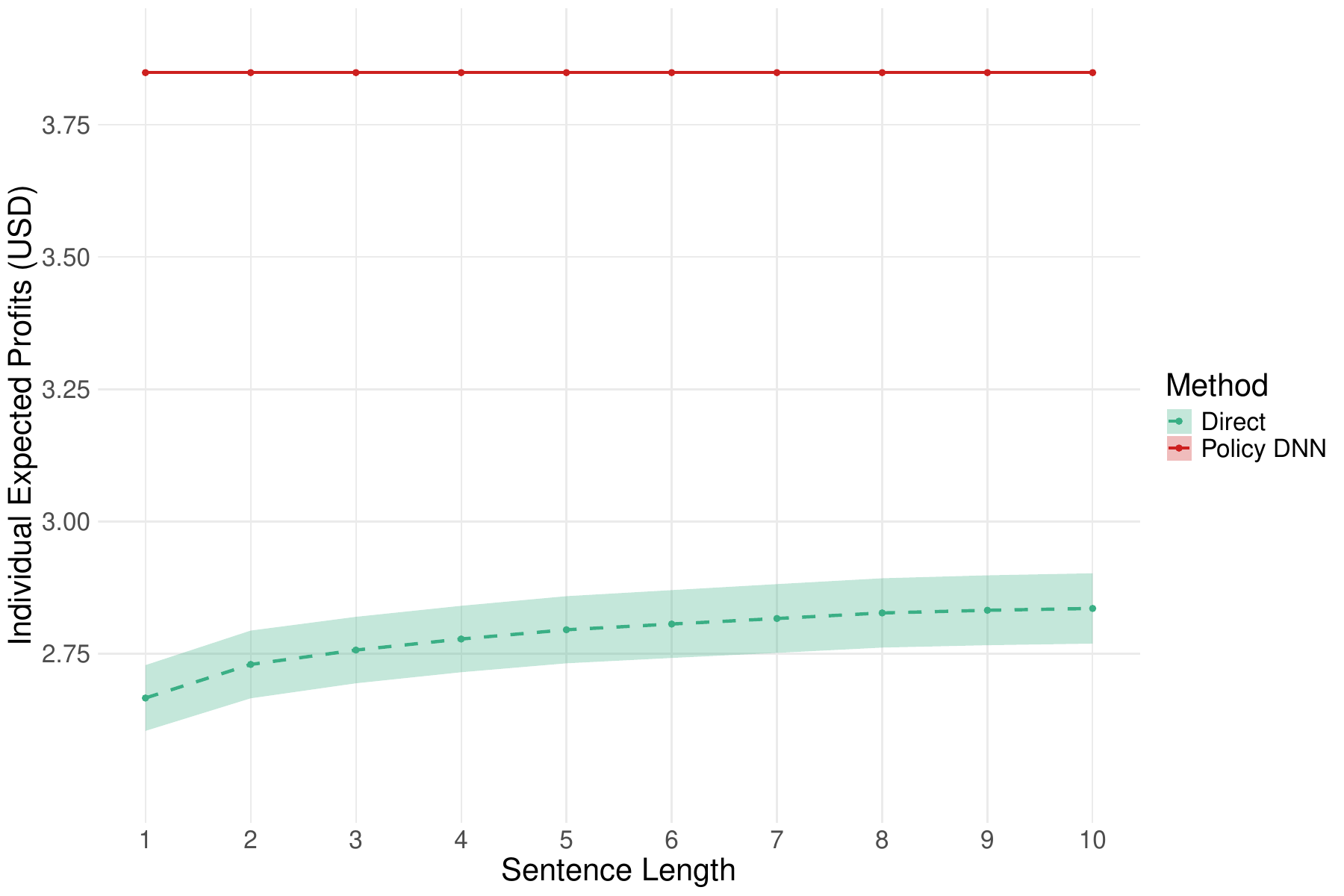}}\medskip{}
\par\end{centering}
\begin{centering}
\subfloat[Out of sample expected profits]{\centering{}\includegraphics[width=0.6\paperwidth]{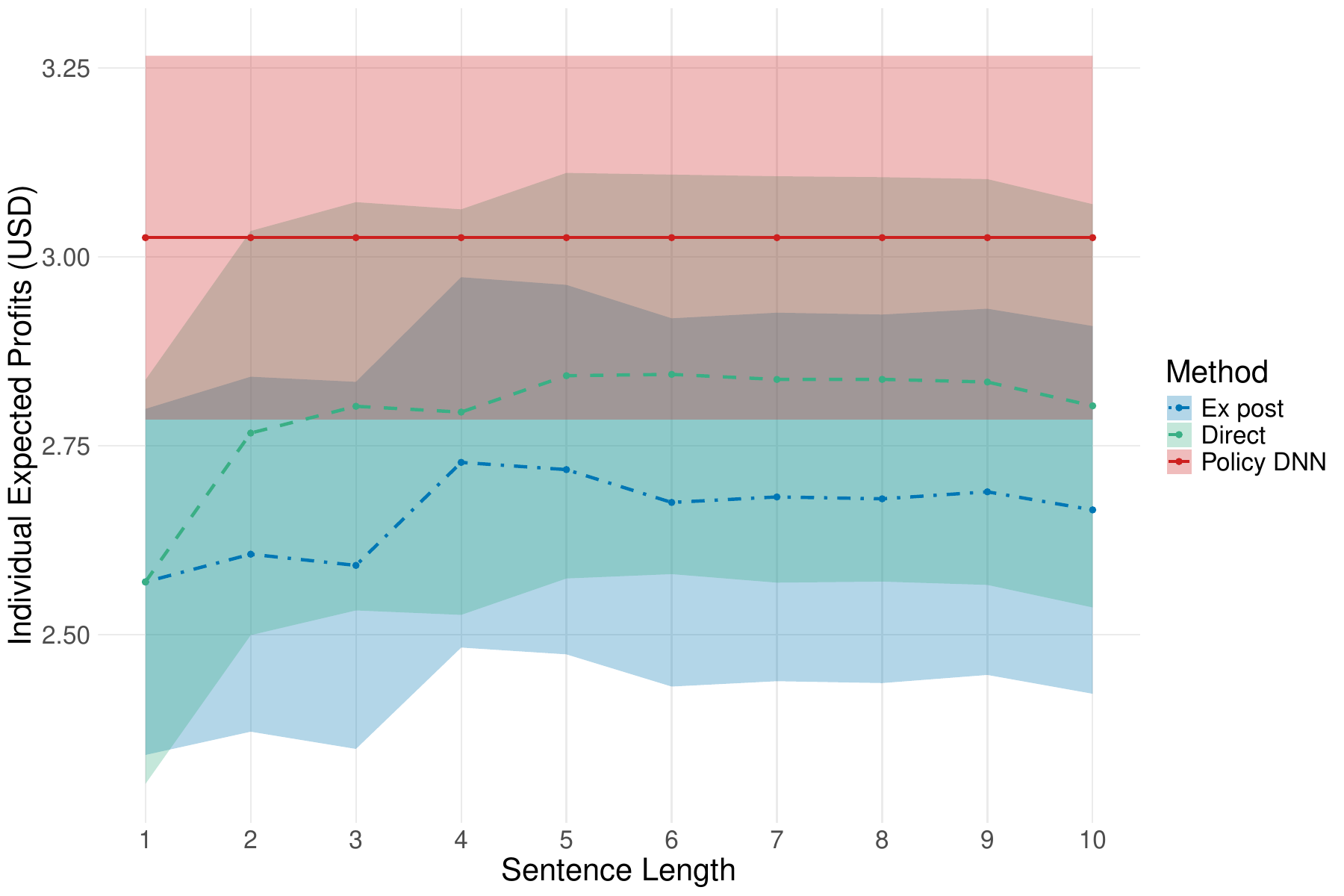}}\medskip{}
\par\end{centering}
{\small Note: These figures plots the expected profits by the number
of clauses. The bands represent one standard error. The upper panel
is the in sample expected profits and the bottom panel is the out
of sample expected profits. In both panels, the }{\small\emph{direct}}{\small{}
method learns the comprehensible policy from the data and the }{\small\emph{ex
post}}{\small{} method projects down the Policy DNN down to a comprehensible
policy. Policy DNN is represented by the dashed horizontal line and
there are no analytical standard errors for Policy DNN in sample. }{\small\par}
\end{figure}

\section*{Tables}

\begin{table}[H]
\begin{centering}
\caption{Individual expected profits (out of sample)\label{tab:Black-Box-Profits}}
\begin{tabular}{c|cc}
Method & Individual Profits & SE\tabularnewline
\hline 
Policy DNN & 3.026 & 0.241\tabularnewline
Causal DNN & 2.919 & 0.279\tabularnewline
Comprehensible Policy & 2.802 & 0.270\tabularnewline
Causal Forest & 2.778 & 0.260\tabularnewline
Policy Tree & 2.717 & 0.251\tabularnewline
Lasso & 2.702 & 0.267\tabularnewline
Blanket & 2.419 & 0.227\tabularnewline
\end{tabular}\medskip{}
\par\end{centering}
{\small Note: These are the out of sample, expected profits for an
individual for different optimal targeting policy $d^{*}(x)$ methods.
The blanket targeting policy represents the expected profits if everyone
was mailed the promotion. The standard errors for the Policy DNN are
computed analytically. The standard errors for the other models are
computed via bootstrap. Policy Tree is a decision tree is learned
by projecting down the Causal Forest and has depth two. The optimal
comprehensible policy uses three clauses and is formed using the }{\small\emph{direct}}{\small{}
method that learns the comprehensible policy from the data.}{\small\par}
\end{table}

\begin{table}[H]
\begin{centering}
\caption{Targeting policy differences\label{tab:Targeting-policy-differences}}
\begin{tabular}{c|cc|c}
 & \multicolumn{2}{c|}{\textbf{Policy DNN}} & \tabularnewline
\textbf{Comprehensible Policy} & Targeted & Not targeted & \tabularnewline
\hline 
Targeted & 9,997 & 18,762 & 28,759\tabularnewline
Not targeted & 19,491 & 110,556 & 130,047\tabularnewline
\hline 
 & 29,488 & 129,318 & 158,806\tabularnewline
\end{tabular}\medskip{}
\par\end{centering}
{\small Note: This is a confusion table for the targeting policy differences
for the optimal three-clause comprehensible policy and for the Policy
DNN policy. The Policy DNN targets 18.1\% of customers and the comprehensible
policy targets 18.3\% of the customers. The two targeting policies
overlap on 75.9\% of customers. The optimal comprehensible policy
is formed using the }{\small\emph{direct}}{\small{} method that learns
the comprehensible policy from the data.}{\small\par}
\end{table}

\begin{table}[H]
\begin{centering}
\caption{Cost of explanation for three-clause comprehensible policies\label{tab:Cost-of-Explanation}}
{\footnotesize{}%
\begin{tabular}{c|cc|cc|c}
 & \multicolumn{2}{c|}{{\footnotesize\textbf{In sample}}} & \multicolumn{2}{c|}{{\footnotesize\textbf{Out of sample}}} & {\footnotesize\textbf{Comprehensible Targeting Policy}}\tabularnewline
 & {\footnotesize Profits} & {\footnotesize CoE} & {\footnotesize Profits} & {\footnotesize CoE} & {\footnotesize Target customer if:}\tabularnewline
\hline 
\multirow{3}{*}{{\footnotesize\textbf{Direct}}} & \multirow{3}{*}{{\footnotesize$\$2.76$}} & \multirow{3}{*}{{\footnotesize$\$1.09$}} & \multirow{3}{*}{{\footnotesize$\$2.80$}} & \multirow{3}{*}{{\footnotesize$23\cent$}} & {\footnotesize high XMAS items (2Y)}\tabularnewline
 &  &  &  &  & {\footnotesize\textbf{and }}{\footnotesize\emph{not}}{\footnotesize{}
high Spring sales (2Y)}\tabularnewline
 &  &  &  &  & {\footnotesize\textbf{or}}{\footnotesize{} low spend last year holiday
promo}\tabularnewline
\hline 
\multirow{3}{*}{{\footnotesize\textbf{Ex post}}} & \multirow{3}{*}{{\footnotesize -}} & \multirow{3}{*}{{\footnotesize -}} & \multirow{3}{*}{{\footnotesize$\$2.59$}} & \multirow{3}{*}{{\footnotesize$53\cent$}} & {\footnotesize high XMAS sales (2Y)}\tabularnewline
 &  &  &  &  & {\footnotesize\textbf{or}}{\footnotesize{} high Back to School items
(2Y)}\tabularnewline
 &  &  &  &  & {\footnotesize\textbf{or}}{\footnotesize{} no two years ago Expo promotion
items }\tabularnewline
\hline 
{\footnotesize\textbf{Policy DNN}} & {\footnotesize$\$3.84$} & {\footnotesize\textendash{}} & {\footnotesize$\$3.03$} & {\footnotesize\textendash{}} & \tabularnewline
\end{tabular}}\medskip{}
\par\end{centering}
{\small Note: This table provides the cost of explanation (CoE) of
the three-clause optimal comprehensible policies and states their
targeting policy. The }{\small\emph{direct}}{\small{} method learns
the comprehensible policy from the data and the }{\small\emph{ex post}}{\small{}
method projects down the Policy DNN down to a comprehensible policy. }{\small\par}
\end{table}
\newpage{}

\appendix

\section*{Appendix}

\section{Comprehensible policies and decision trees\label{sec:Decision-Trees-and-Sentences}}

In this section, I demonstrate that comprehensible policies are subsets
of decision trees and show a comprehensible policy with a fixed number
of clauses has a finite Vapnik-Chervonenkis dimension. In Section
\ref{sec:Comprehensible-class}, I constructed the comprehensible
policy class to be targeting policies that can be represented by a
sentence. A decision tree of depth $\ell$ is typically more complex
than a comprehensible policy of $\ell$ clauses; a sentence of $\ell$
clauses can be represented by a tree of depth $\ell$ by not vice
versa. To show this, I visualize how to construct the decisions trees
from comprehensible policies for one to three clauses and the provide
an algorithm to generally do so. I then leverage the link from comprehensible
policies to decisions trees to control the Vapnik\textendash Chervonenkis
(VC) dimension of comprehensible policies.
\begin{lem}
\label{lem:Comprehensible-Sentence-Decision-Tree} A comprehensible
policy of length $\ell$ can be represented by a full decision tree
of depth $\ell$.
\end{lem}

\begin{proof}
I provide a proof by construction. The $\ell=1$ case is straightforward
and I show the representation explicitly for a comprehensible policy
of $\ell=2$ clauses and $\ell=3$ clauses. Then, I present algorithms
to map a comprehensible policy of length $\ell$ to a decision tree
of depth $\ell$. 

A one clause comprehensible policy can be represented as a decision
tree of depth one which has one split. Figure \ref{fig:Comprehensible-policies-and-decision-trees-2}
and Figure \ref{fig:Comprehensible-policies-and-decision-trees-3}
provide the mapping between a comprehensible policy of $\ell=2$ clauses
and $\ell=3$ clauses to their respective decision trees. A key result
in these two examples is that a two clauses comprehensible policy
can be represented by a decision tree of depth two and a three clauses
comprehensible policy can be represented by a decision tree of depth
three. 

To complete the construction, Algorithm \ref{alg:Add-and}, \ref{alg:Add-or},
and \ref{alg:Add-xor} show how to grow the decision tree when ``and'',
``or'', and ``xor'' operators and a clause are added a comprehensible
policy. Adding one additional clause the the comprehensible policy
increases the decision tree by an extra level of depth. Thus, a comprehensible
policy of length $\ell$ can be represented by a decision tree of
depth $l$.
\end{proof}
\begin{cor}
\label{cor:Tree-to-Sentence}A decision tree of depth $\ell$ cannot
always be represented by a comprehensible policy of length $\ell$.
\end{cor}

\begin{proof}
Proof by counterexample. A decision tree of depth $2$ can use three
different clauses. A comprehensible policy of length 2 can only use
two clauses.
\end{proof}
\begin{lem}
\label{lem:Comprehensible-Sentence-Finite-VC-Dim}A comprehensible
policy of finite length $\ell\in\mathbb{N}$ has a finite Vapnik\textendash Chervonenkis
(VC) dimension for finite number of covariates $p\in\mathbb{N}$.
\end{lem}

\begin{proof}
\citet{Athey2021} supply the asymptotic VC dimension for a decision
tree of $l$ layers, $p$ covariates, and observations $n$ as $VC(d_{DT})=\tilde{\mathcal{O}}(2^{l_{n}}\log_{2}(n))$
where $l_{n}=\lfloor\kappa\log_{2}(p)\rfloor,\kappa<1/2$. Here, $f(n)=\tilde{\mathcal{O}}(g(n))$
implies there is a function that scales polylogarithmically in its
arguments for $f(n)<h(g(n))g(n)$ where $g(n)$ is at least polynomial
in $n$. 

Lemma \ref{lem:Comprehensible-Sentence-Decision-Tree} shows that
a comprehensible policy of length $\ell$ can be represented by a
decision tree of depth $\ell$. This result implies that the VC dimension
of the full decision tree of depth $\ell$ will be an upper bound
for that of the comprehensible policy, 
\[
VC(d_{\text{comp}})\leq VC(d_{DT})=\tilde{\mathcal{O}}(2^{l_{n}}\log_{2}(p)).
\]

However since the number of clauses (and depth of the decision tree)
$\ell$ is finite in my setting, I choose $l_{n}'=\min\{l_{n},\ell\}\leq\ell$
for the upper bound of the decision tree's depth. This choice of the
model structure breaks the dependence of the VC dimension of the tree
to $n$ as $2^{l}\log_{2}(p)=C$ for finite $p$. Then, I see that
$VC(d_{\text{comp}})\leq VC(d_{DT})\leq C'$ where $d_{\text{comp}}$
has $\ell$ clauses and $d_{DT}$ has depth $\ell$. Thus, a comprehensible
sentence of finite length $\ell$ has finite VC dimension.
\end{proof}

\section{Supporting proofs for Section \ref{sec:Ex-post-policy} \label{sec:Proofs-for-XAI}}

I provide the proof for Theorem \ref{thm:Convergence-rate-of-ex-post-comprehensible-policies}
in Section \ref{sec:Ex-post-policy}. I restate the theorem and related
definitions for ease of exposition. I define $\pi(d)$ to be the individual-level
profits from targeting policy $d$, and then define
\begin{align*}
\check{d}^{*}(x) & =\arg\max_{d\in\mathcal{F}_{\text{comp}}}E_{P}\left[\left|\pi(d_{BB}^{*})-\pi(d)\right|\right]\\
\hat{d}_{ex}(x) & =\arg\min_{d\in\mathcal{F}_{\text{comp}}}E_{n}\left[\left|\pi(d_{BB}^{*})-\pi(d)\right|\right],
\end{align*}
where the first line has the expectation taken over the population
distribution $P$ and the second line has the expectation taken over
the sample analog. The \emph{ex post} optimal comprehensible policy
is $\hat{d}_{ex}(x)$ and the black box policy is $d_{BB}^{*}(x)$.
I further define $\Gamma(d)=E_{P}\left[\left|\pi(d_{BB}^{*})-\pi(d)\right|\right]$
to be the profit loss for the targeting rule $d$ evaluated in the
population and $\Gamma_{n}(d)=E_{n}\left[\left|\pi(d_{BB})-\pi(d)\right|\right]$
to be the profit loss for targeting rule $d$ in the sample. I suppress
the dependence on $x$ for notational simplicity.
\begin{thm*}
(Uniform convergence rate of ex post comprehensible policies) Under
Assumptions \ref{assu:Unconfoundedness}, \ref{assu:Overlap}, and
\ref{assu:SUTVA} and the assumption that the outcome variable ($Y$)
is bounded, for a hypothesis space of comprehensible policies with
$\ell$ clauses ($\mathcal{H}=\mathcal{F}_{\text{comp}}^{\ell}$)
that has bounded VC dimension ($VC(\mathcal{H})<V<\infty$),
\[
\sup_{P}E_{P}\left[\Gamma(\check{d}^{*})-\Gamma(\hat{d}_{ex})\right]\leq C\sqrt{\frac{V}{n}}=O_{p}\left(1/\sqrt{n}\right).
\]
\end{thm*}
\begin{proof}
I extend the analogous proof of Theorem 2.1 in \citet{Kitagawa2018}
to provide the upper bound. First, I see that 
\begin{align}
\Gamma(\check{d})-\Gamma(\hat{d}_{ex}) & =\Gamma(\check{d})-\Gamma_{n}(\hat{d}_{ex})+\Gamma_{n}(\hat{d}_{ex})-\Gamma(\hat{d}_{ex})\nonumber \\
 & \leq\Gamma(\check{d})-\Gamma_{n}(\hat{d}_{ex})+\sup_{d\in\mathcal{F}_{\text{comp}}^{\ell}}\left|\Gamma_{n}(d)-\Gamma(d)\right|\label{eq:Thm-3-Derivation}\\
 & \leq\Gamma(\check{d})-\Gamma_{n}(\check{d})+\sup_{d\in\mathcal{F}_{\text{comp}}^{\ell}}\left|\Gamma_{n}(d)-\Gamma(d)\right|\nonumber \\
 & \leq2\sup_{d\in\mathcal{F}_{\text{comp}}^{\ell}}\left|\Gamma_{n}(d)-\Gamma(d)\right|\nonumber 
\end{align}
where used that $\Gamma_{n}(\hat{d}_{ex})\geq\Gamma_{n}(\check{d})$
from the optimality of $\hat{d}_{ex}$ to get to the third line. Since
the bound holds for all $\Gamma(\check{d})$, it also holds for $\Gamma(\check{d}^{*})$.
Thus, I have

\begin{equation}
\Gamma(\check{d}^{*})-\Gamma(\hat{d}_{ex})\leq2\sup_{d\in\mathcal{F}_{\text{comp}}^{\ell}}\left|\Gamma_{n}(d)-\Gamma(d)\right|.\label{eq:Thm-3-Derivation-Second-Line}
\end{equation}

I now need to bound the empirical process term $\sup_{d\in\mathcal{F}_{\text{comp}}^{\ell}}|\Gamma_{n}(d)-\Gamma(d)|$
to complete the proof. From the Assumptions \ref{assu:Unconfoundedness},
\ref{assu:Overlap}, and \ref{assu:SUTVA} and the bounded outcome
variable $(Y)$ assumption, the $\Gamma(d)$ function is bounded with
$||\Gamma(d)||_{\infty}<\bar{\Gamma}$. Assumption \ref{assu:Overlap}
provides strict overlap (for propensity score $e(x)$, I have $\epsilon<e(x)<1-\epsilon,\forall x$
and for some $\epsilon>0$). Lastly, the hypothesis space of comprehensible
policies with a fixed length $l$ has a bounded VC dimension $VC(\mathcal{H})<V<\infty$
from Lemma \ref{lem:Comprehensible-Sentence-Finite-VC-Dim}. 

I then use Lemma A.4 in the supplement of \citet{Kitagawa2018} to
bound the the empirical process term. The lemma adapted to my notation
provides 
\[
E_{P}\left[\sup_{d\in\mathcal{F}_{\text{comp}}^{\ell}}\left|\Gamma_{n}(d)-\Gamma(d)\right|\right]\leq C_{1}\bar{\Gamma}\sqrt{\frac{V}{n}}
\]
for some constant $C_{1}$. Applying this result to Equation \ref{eq:Thm-3-Derivation-Second-Line},
I have that 
\[
E_{p}\left[\Gamma(\check{d})-\Gamma(\hat{d}_{ex})\right]\leq C\sqrt{\frac{V}{n}}.
\]
Then taking the supremum over the possible probability distributions
$P$ that satisfy the assumptions yields,
\[
\sup_{P}E_{p}\left[\Gamma(\check{d})-\Gamma(\hat{d}_{ex})\right]\leq C\sqrt{\frac{V}{n}}=O_{p}(1/\sqrt{n})
\]
as $V$ and $C$ do not depend on $n$.
\end{proof}

\section*{Online Appendix}

\section{Greedy vs. brute force solutions \label{sec:Bounding-the-greedy-algorithm}
}

This section compares the two solutions for finding the optimal comprehensible
policy. I describe worst case bounds for the greedy algorithm in Section
\ref{subsec:Worst-case-bounds} and report the brute force solution's
results for the empirical application in Section \ref{subsec:Brute-force-algorithm-results}.

\subsection{Worst-case bounds for the greedy algorithm\label{subsec:Worst-case-bounds}}

In this section, I construct a worst-case bound for the greedy algorithm
compared to the brute force algorithm by extending results from \citet{Nemhauser1978}
and \citet{Bian2017}. I first rewrite the inverse propensity weighted
(IPWE) profit estimator function to separate the constant from the
term that depends on the policy function $d$,
\begin{align*}
\hat{\Pi}(d) & =\sum_{i=1}^{n}\frac{1-W_{i}}{1-e(x_{i})}\pi_{i}(0)(1-d(x_{i}))+\frac{W_{i}}{e(x_{i})}\pi_{i}(1)d(x_{i})\\
 & =\sum_{i=1}^{n}\frac{1-W_{i}}{1-e(x_{i})}\pi_{i}(0)+\sum_{i=1}^{n}\underbrace{\left(\frac{W_{i}}{e(x_{i})}\pi_{i}(1)-\frac{1-W_{i}}{1-e(x_{i})}\pi_{i}(0)\right)}_{=\omega_{i}}d(x_{i}),
\end{align*}
where $\omega_{i}$ is defined as the weight on $d$ in the IPWE profits
function. The weight represents the difference in profits if individual
$i$ was targeted or not.

I now make the relation of the policy to the targeting sentence explicit.
I define sentence $S$ to be a set of clauses plus logic operators
and the policy function to be $d_{S}(x_{i})$. Maximizing the IPWE
profits to the policy function is equivalent to maximizing the following
objective plus a constant,
\begin{equation}
F(d_{S})=\sum_{i=1}^{n}\omega_{i}d{}_{S}(x_{i}).\label{eq:Worst-bound-Objective-Function}
\end{equation}
I define the each clause (e.g., ``lives in Chicago'') as $v$ and
further define the coverage set $M_{v}=\{i:\text{clause \ensuremath{v} is true for \ensuremath{x_{i}}}\}$. 

First, I consider the specific conditions where $F(d_{S})$ is monotone
and submodular to apply the $(1-1/e)$ worst-case bound from \citet{Nemhauser1978}.
There are two sufficient conditions:
\begin{enumerate}
\item Only or logic operators are used. The sentences targets people who
are covered by any clause in the sentence, so $d_{S}(x_{i})=\mathbf{1}\{i\in\cup_{v\in S}M_{v}\}$.
The other possible logic operators (and, xor, and not) are not considered. 
\item Only non-negative weights ($\omega_{i}\geq0$) are considered.
\end{enumerate}
Under the two sufficient conditions the objective function (Equation
\ref{eq:Worst-bound-Objective-Function}) is a weighted coverage function
and is both monotone and submodular. \citet{Nemhauser1978} implies
in the worst case, the greedy algorithm will recover $(1-1/e)\approx63.2\%$
of the profits of the brute force solution. However, these sufficient
conditions are quite restricting in practice and the second assumption
cannot be verified in the data because both potential outcomes are
not observed for each individual.

Second, I discuss the general case where I optimize over all possible
clauses, include the full set of logic operators, and allow all values
of $\omega_{i}$. \citet{Bian2017} offer a computable bound that
depends on the submodularity of the objective function, 
\[
F(d_{S}^{Greedy})\geq\frac{1}{\alpha}\left(1-e^{-\alpha\gamma}\right)F(d_{S}^{BruteForce}).
\]
The two parameters of the bound are the submodularity ratio $\gamma$
and the curvature of the objective function $\alpha$. Note that when
$\gamma=\alpha=1$, the bound recovers that of \citet{Nemhauser1978}.

I follow the procedure in \citet{Bian2017} and compute the bound
for the optimal comprehensible policy with three clauses. I find $\hat{\alpha}=1.00$
and $\hat{\gamma}=0.714$, indicating the profit objective is relatively
close to being submodular and monotone. These values yield a worst-case
bound of $51.1\%$, meaning the greedy algorithm recovers at least
$51.1\%$ of the brute force solution's profits. This compares relatively
favorably to the $63.2\%$ theoretical worst-case bound that applies
only when objectives are strictly monotone and submodular. The results
suggest that greedy optimization remains effective.

\subsection{Brute force algorithm empirical results \label{subsec:Brute-force-algorithm-results}}

In this section, I report the brute force solutions for the empirical
application. First, the one clause optimal comprehensible policy is
identical to the greedy algorithm, so there is no gap in performance. 

For two clauses, the greedy algorithm yields \$2.730 in sample and
\$2.767 out of sample for individual expected profits. The brute
force algorithm yields \$2.733 in sample and \$2.684 out of sample.
In sample, the brute force algorithm's solution is the upper bound
and the gap between the greedy and brute force solutions is less than
a cent per person. Out of sample, we see that the brute force solution
doing slightly worse than the greedy solution but the difference is
not statistically significant (Figure \ref{fig:Individual-Expected-Profits}
provides the greedy solution's standard errors bands). Together, these
results suggest the greedy solution is doing quite well and achieving
close to the best two clause comprehensible targeting policy. 

For the empirical application, brute force solutions for finding optimal
comprehensible targeting policies with three or more clauses are computationally
intractable because the combinatorial space is too large. For the 150
covariates discretized to three clauses each, there are $6.5\times10^{9}$
combinations to search over for an optimal comprehensible policy with
three clauses.

\section{Profit loss function derivation \label{sec:Profit-loss-derivation}}

This section derives the profit loss function used in Section \ref{sec:Ex-post-policy}.
The individual-level inverse propensity weighted (IPWE) profit loss
estimator for policy function $d(x_{i})$ is

\begin{equation}
\pi(d(x_{i}))=\frac{1-W_{i}}{1-e(x_{i})}\pi_{i}(0)(1-d(x_{i}))+\frac{W_{i}}{e(x_{i})}\pi_{i}(1)d(x_{i}).\label{eq:Individual-Profit-Estimator}
\end{equation}
In this notation, the sample profits are then $\hat{\Pi}(d)=\sum_{i=1}^{n}\pi(d(x_{i}))$
with is an estimate for average population profits $\Pi(d)=E[\pi(d(x_{i}))]$.

The individual-level IPWE of profit difference of two policies $d(x_{i}),d'(x_{i})$
is

\begin{align}
\pi(d(x_{i}))-\pi(d'(x_{i})) & =\left(\frac{1-W_{i}}{1-e(x_{i})}\pi_{i}(0)(1-d(x_{i}))+\frac{W_{i}}{e(x_{i})}\pi_{i}(1)d(x_{i})\right)\label{eq:IPWE-difference-two-policies}\\
 & ~~-\left(\frac{1-W_{i}}{1-e(x_{i})}\pi_{i}(0)(1-d'(x_{i}))-\frac{W_{i}}{e(x_{i})}\pi_{i}(1)d'(x_{i})\right)\nonumber \\
 & =\frac{1-W_{i}}{1-e(x_{i})}\pi_{i}(0)(1-d(x_{i})-(1-d'(x_{i})))+\frac{W_{i}}{e(x_{i})}\pi_{i}(1)(d(x_{i})-d'(x_{i}))\nonumber \\
 & =\frac{1-W_{i}}{1-e(x_{i})}\pi_{i}(0)(d'(x_{i})-d(x_{i}))+\frac{W_{i}}{e(x_{i})}\pi_{i}(1)(d(x_{i})-d'(x_{i}))\nonumber \\
 & =\left(d(x_{i})-d'(x_{i})\right)\left(\frac{W_{i}}{e}\pi_{i}(1)-\frac{1-W_{i}}{1-e(x_{i})}\pi_{i}(0)\right)\nonumber 
\end{align}

To construct a loss function, I consider the absolute difference of
the individual-level IPWE profit differences,
\begin{align*}
|\pi(d(x_{i}))-\pi(d'(x_{i}))| & =\left|(d(x_{i})-d'(x_{i}))(\frac{W_{i}}{e(x_{i})}\pi_{i}(1)-\frac{1-W_{i}}{1-e(x_{i})}\pi_{i}(0))\right|\\
 & =\mathbf{1}\{d(x_{i})\neq d'(x_{i})\}\left|\frac{W_{i}}{e(x_{i})}\pi_{i}(1)-\frac{1-W_{i}}{1-e(x_{i})}\pi_{i}(0)\right|
\end{align*}
where I used the fact that $d(x_{i})$ and $d'(x_{i})$ are indicator
functions to get to the second line. I treat the sample-level difference
as the loss function of interest for two policies $d,d'$
\begin{align*}
\mathcal{L}(d,d') & =|\hat{\Pi}(d)-\hat{\Pi}(d')|\\
 & =\sum_{i=1}^{n}|\pi(d(x_{i}))-\pi(d'(x_{i}))|\\
 & =\sum_{i=1}^{n}\mathbf{1}\{d(x_{i})\neq d'(x_{i})\}\left|\frac{W_{i}}{e(x_{i})}\pi_{i}(1)-\frac{1-W_{i}}{1-e(x_{i})}\pi_{i}(0)\right|.
\end{align*}

\section{Black box implementation details \label{sec:Implementation-Details}}

In this section, I discuss the implementation details for the black
box benchmarks in the empirical application. 

\textbf{Policy DNN}. I implement a DNN architecture of three hidden
layers each with $12$ rectified linear units (ReLU) nodes. I use
the ADAM optimizer with a learning rate of $0.001$ and a scheduler
with the exponential decay parameter of $0.9$ that updates every
200 epochs. I train the Policy DNN for $8700$ epochs. The surrogate
policy function used is $f(z)=\frac{\tanh(z+1)}{2}$. I use 80\% of
the training data to update the DNN and use 20\% of the training data
as the validation data to evaluate the performance of the DNN. Policy
DNN is implemented using the \texttt{torch} package in R.

\textbf{Causal DNN}. I implement a DNN architecture of two hidden
layers each with $30$ ReLU nodes and a dropout rate of $0.5$ for
each layer. I use the ADAM optimizer with a learning rate of $0.0005$
and a scheduler with the exponential decay parameter of $0.9$ that
updates every 200 epochs. I train Causal DNN for $2000$ epochs. I
use 80\% of the training data to update the DNN and use 20\% of the
training data as the validation data to evaluate the performance of
the DNN. Causal DNN is implemented using the \texttt{torch} package
in R.

\textbf{Causal Forest}. I use the \texttt{grf} package in R to implement
the Causal Forest and Policy Tree \citep{Athey2019}. I use $2000$
trees for the Causal Forest.

\textbf{Lasso}. The Lasso is implemented with ten-fold cross validation.
I use the \texttt{gamlr} package in R to implement the Lasso. 

\newpage{}

\section*{Algorithms}

\begin{algorithm}
\caption{Adding an \textquotedblleft and\textquotedblright{} operator and a
clause to a decision tree \label{alg:Add-and}}
\medskip{}

\textbf{Objective}: Add ``and D'' to the comprehensible policy.

\textbf{Setup}: ``and'' is the logic operator and ``D'' is the
new clause being added. Take the decision tree representation of the
comprehensible policy without the addition.
\begin{enumerate}
\item For all terminal nodes in the decision tree that have value $1$,
grow a split on clause D and that have terminal node value $1$ if
clause D is true and $0$ otherwise
\end{enumerate}
\end{algorithm}

\begin{algorithm}
\caption{Adding an \textquotedblright or\textquotedblright{} operator and a
clause to a decision tree\label{alg:Add-or}}

\medskip{}
\textbf{Objective}: Add ``or D'' to the comprehensible policy.

\textbf{Setup}: ``or'' is the logic operator and ``D'' is the
new clause being added. Take the decision tree representation of the
comprehensible policy without the addition.
\begin{enumerate}
\item For all terminal nodes in the decision tree that have value $0$,
grow a split on clause D and that have terminal node value $1$ if
clause D is true and $0$ otherwise
\end{enumerate}
\end{algorithm}

\begin{algorithm}
\caption{Adding a \textquotedblright xor\textquotedblright{} operator and a
clause to a decision tree\label{alg:Add-xor}}

\medskip{}
\textbf{Objective}: Add ``xor D'' to the comprehensible policy.

\textbf{Setup}: ``xor'' is the logic operator and ``D'' is the
new clause being added. Take the decision tree representation of the
comprehensible policy without the addition.
\begin{enumerate}
\item For all terminal nodes in the decision tree that have value $0$,
grow a split on clause D and that have terminal node value $1$ if
clause D is true and $0$ otherwise
\item For all terminal nodes in the decision tree that have value $1$,
grow a split on clause D and that have terminal node value $0$ if
clause D is true and $0$ otherwise
\end{enumerate}
\end{algorithm}

\newpage{}

\section*{Figures}

\begin{figure}[H]
\begin{centering}
\caption{Propensity score density\label{fig:Prop-Score-Overlap}}
\includegraphics[width=0.45\paperwidth]{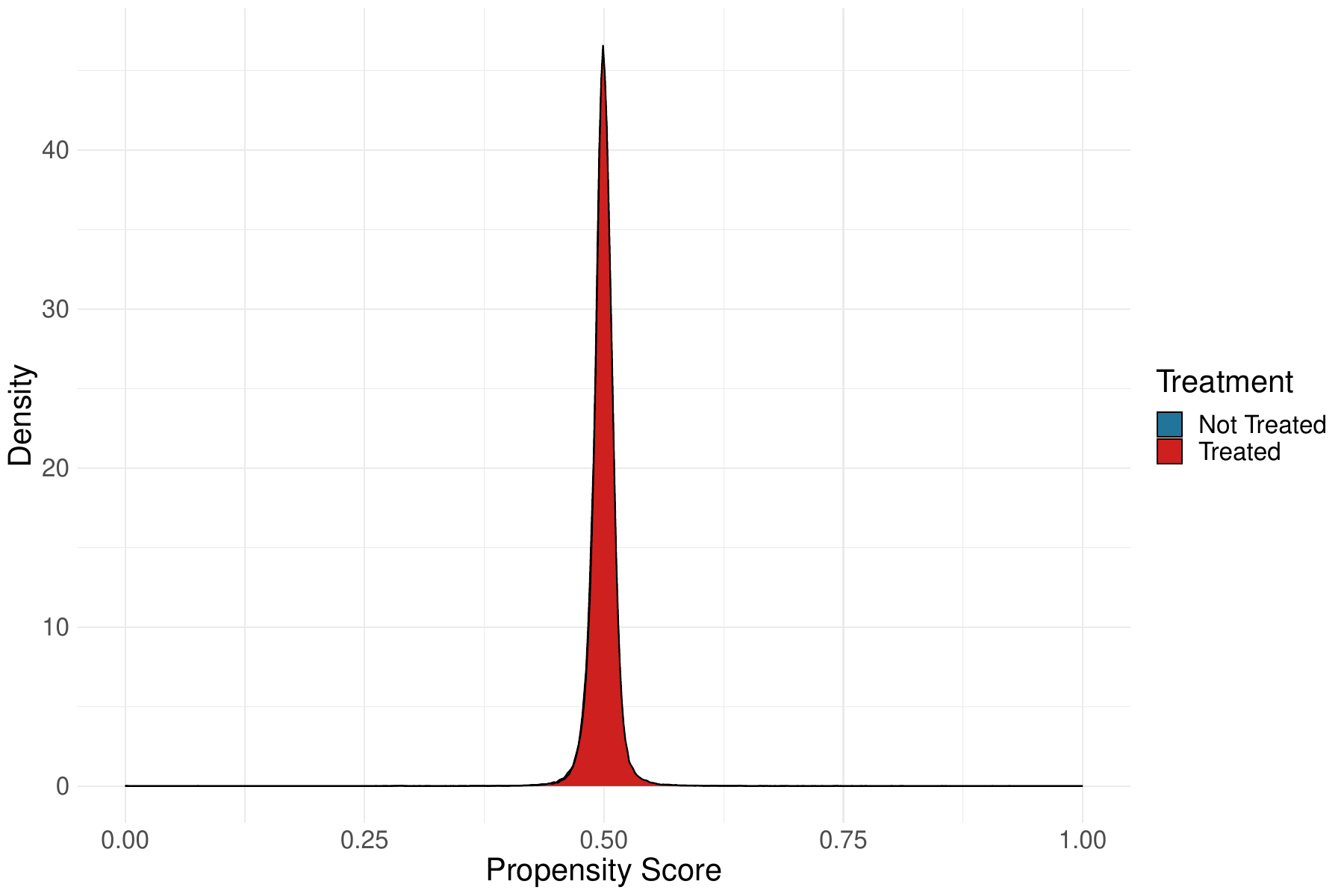}
\par\end{centering}
{\small Note: This figure plots the density of the predicted propensity
scores for the treated and not treated groups. The propensity score
model was estimated using a logistic regression of the treatment indicator
on the RFM covariates in the data set. The two densities essentially
fully overlap with one another, and these results suggest the overlap
assumption holds in the data.}{\small\par}
\end{figure}

\begin{figure}[H]
\caption{Two clause comprehensible policies and decision trees ($\ell=2)$\label{fig:Comprehensible-policies-and-decision-trees-2}}

\begin{centering}
\includegraphics[width=0.8\textwidth,page=1]{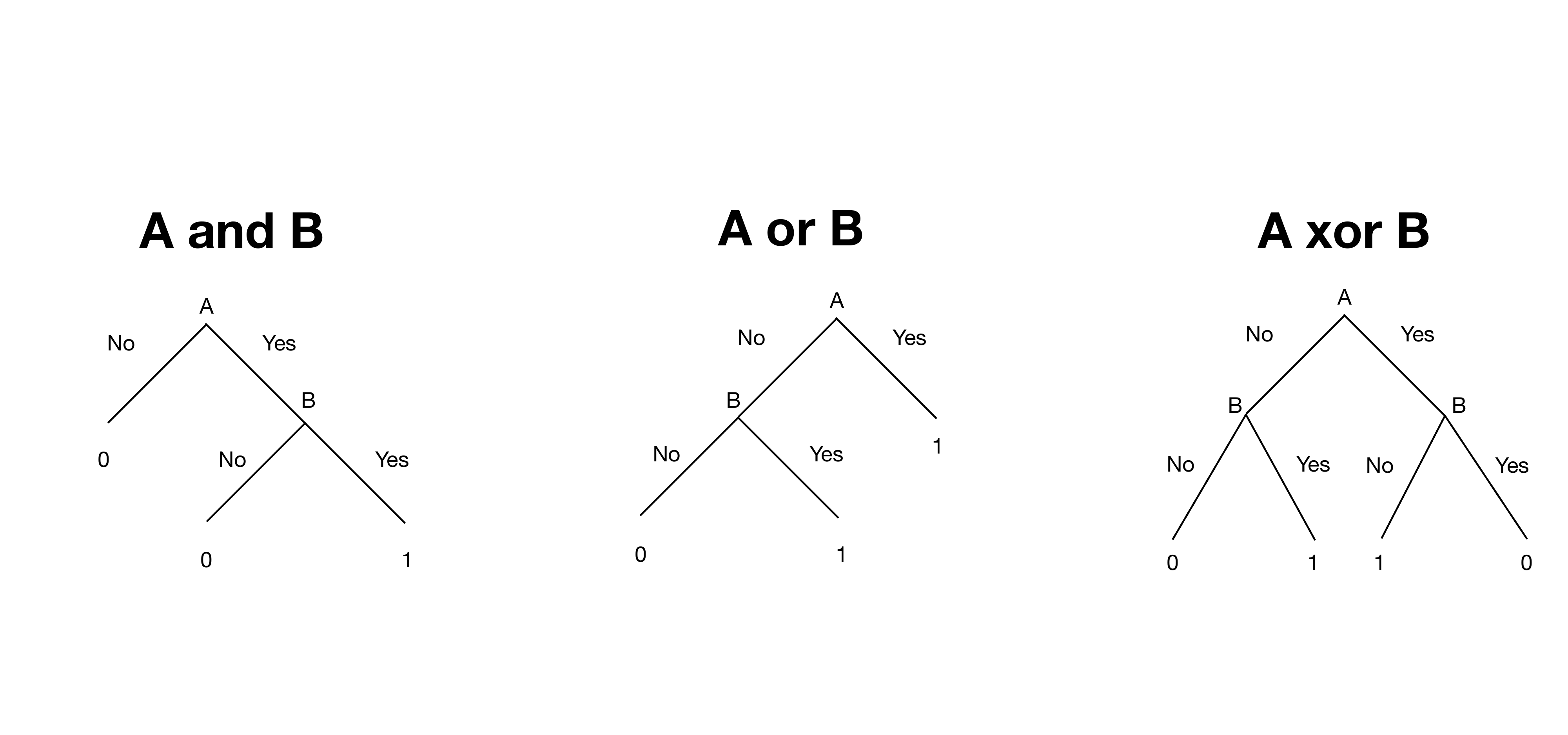}
\par\end{centering}
{\small Note: This figure demonstrates the mapping between a comprehensible
policy with two clauses (A, B) to a decision tree of depth two. The
comprehensible policy on the left states target if ``A and B'' in
which a customer is targeted if the clause A and clause B are both
true and not targeted otherwise. The corresponding decision tree first
has a split whether clause A is true (yes/no) and then conditional
on clause A being true it has another split on whether clause B is
true. The terminal nodes of $1$ indicate targeted and $0$ indicated
not targeted. The comprehensible policy and its respective decision
tree will target and not target the same sets of customers. The other
comprehensible policies and decision trees in the figure are mapped
similarly.}{\small\par}
\end{figure}

\begin{figure}[H]
\caption{Three clause comprehensible policies and decision trees ($\ell=3)$\label{fig:Comprehensible-policies-and-decision-trees-3}}

\begin{centering}
\includegraphics[width=1\textwidth,page=1]{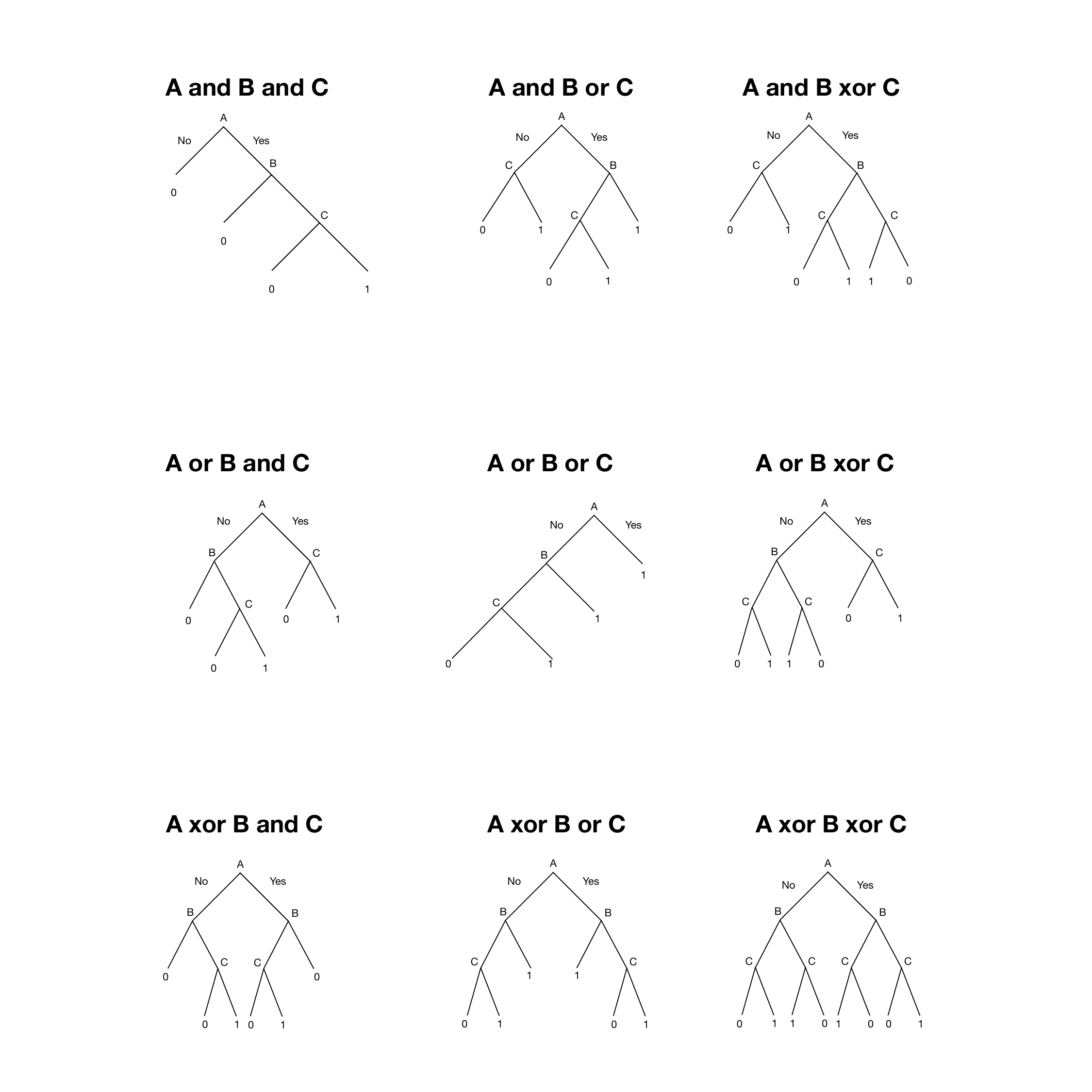}
\par\end{centering}
{\small Note: This figure shows the mapping between a comprehensible
policy with three clauses (A, B, C) to a decision tree of depth three.
The comprehensible policy on the top left states target if ``A and
B and C'' in which a customer is targeted if the clause A, clause
B, and clause C are all true and is not targeted otherwise. The corresponding
decision tree first has a split whether clause A is true (yes/no),
then conditional on clause A being true it has another split on whether
clause B is true, and lastly conditional on A and B being true it
has a split on whether clause C is true. The terminal nodes of $1$
indicate targeted and $0$ indicated not targeted. The comprehensible
policy and its respective decision tree will target and not target
the same sets of customers. The other comprehensible policies and
decision trees in the figure are mapped similarly.}{\small\par}
\end{figure}

\end{document}